\documentclass[11pt]{article}
\usepackage{amsmath,amssymb,amsfonts,amsthm,epsfig}
\usepackage[usenames,dvipsnames]{xcolor}
\usepackage{bm,xspace}
\usepackage{tcolorbox}
\usepackage{cancel}
\usepackage{fullpage}
\usepackage{liyang}
\usepackage{framed}
\usepackage{verbatim}
\usepackage{enumitem}
\usepackage{array}
\usepackage{multirow}
\usepackage{afterpage}
\usepackage{mathrsfs}
\usepackage{pifont} 
\usepackage{chngpage}
\usepackage[normalem]{ulem}
\usepackage{boxedminipage}
\usepackage{caption}
\usepackage{subcaption}
\usepackage{bold-extra}
\usepackage{forest}
\usepackage{etoolbox}
\usepackage{thm-restate}
\usepackage{microtype}

\usepackage{algorithm}
\usepackage{algorithmicx}
\usepackage{algpseudocode}

\usepackage{tikz}
\usetikzlibrary{calc,through,backgrounds,decorations.pathreplacing, calligraphy,arrows.meta}
\usetikzlibrary{positioning,chains,fit,shapes}
\usetikzlibrary{patterns}

\usepackage{tikz-cd}

\usepackage{pgfplots}
\pgfplotsset{width=8cm,compat=newest}

\def\colorful{0}

\ifnum\colorful=1
\newcommand{\violet}[1]{{\color{violet}{#1}}}

\fi
\ifnum\colorful=0
\newcommand{\violet}[1]{{{#1}}}

\fi

\newcommand{\YES}{\mathsf{YES}}
\newcommand{\NO}{\mathsf{NO}}

\newcommand{\Bad}{\mathsf{Bad}}
\newcommand{\Good}{\mathsf{Good}}
\newcommand{\Acc}{\mathsf{Acc}}
\newcommand{\Rej}{\mathsf{Rej}}
\newcommand{\pparagraph}[1]{\bigskip \noindent {\bf {#1}}}
\newcommand{\quads}{\mathrm{quad}}
\newcommand{\dt}{d_{\mathrm{Tal}}}
\newcommand{\dtal}{\dt}
\newcommand{\BQP}{\mathsf{BQP}}
\newcommand{\PromiseBQP}{\mathsf{PromiseBQP}}
\newcommand{\BPP}{\mathsf{BPP}}
\newcommand{\PromiseBPP}{\mathsf{PromiseBPP}}
\newcommand{\PromiseQNC}{\mathsf{PromiseQNC}}
\newcommand{\QNC}{\mathsf{QNC}}
\newcommand{\PromiseQuasi}{\mathsf{PromiseBPTIME}(n^{\polylog(n)})}

\newcommand{\PromiseHeurBPP}{\mathsf{PromiseHeurBPP}}
\newcommand{\HeurBPP}{\mathsf{HeurBPP}}

\newcommand{\FindHeavy}{\mathsf{FindHeavy}}
\newcommand{\ApproxExpectation}{\mathsf{ApproxExpectation}}
\newcommand{\ApproxQueryWeights}{\mathsf{ApproxQueryWeights}}

\DeclareMathOperator*{\argmax}{arg\,max}

\Crefname{theorem}{Theorem}{Theorems}
\Crefname{thmenumi}{Theorem}{Theorems}
\AtBeginEnvironment{theorem}{%
    \crefalias{enumi}{thmenumi}%
    \setlist[enumerate,1]{
        label={\textit{(\roman*)}},
        ref={\thetheorem(\roman*)}
    }%
}

\begin{document}

\title{

Quantum Speedups Require Structure or Depth
\vspace{15pt}
}

\author{ 
Guy Blanc \vspace{6pt} \\ 
\hspace{-7pt} {\sl Stanford} \and 
Jordan Docter \vspace{6pt} \\ 
\hspace{-10pt} { {\sl Stanford}} \and 
 Carmen Strassle \vspace{6pt} \\
 \hspace{-5pt} { {\sl Stanford}} \vspace{15pt}
 \and 
Li-Yang Tan \vspace{6pt}  \\
\hspace{-10pt} {{\sl Stanford}}
}

\date{\small{\today}}

 \maketitle

\newcommand{\Simconj}{{Simulation conjecture}}
\newcommand{\simconj}{{simulation conjecture}}

\begin{abstract}
One of the most basic conjectures in quantum complexity theory states that every $t$-query quantum  algorithm can be simulated on most inputs by a  $\poly(t)$-query classical algorithm.  If true, this would provide broad justification for the need for structure in quantum speedups.

  We settle this conjecture for {\sl parallel} quantum algorithms, showing that every $t$-query $d$-round quantum algorithm can be simulated on most inputs with $t^{O(d^2)}$ classical queries. %
    This suggests that for unstructured problems, superpolynomial speedups would require quantum circuits of superconstant depth, and exponential speedups would further require polynomial depth.  In contrast, most known speedups for structured problems are achieved by highly parallel, low-depth algorithms.

Our techniques also carry new implications for the status of $\BPP$ vs.~$\mathsf{BQP}$ relative to a random oracle, a similarly longstanding problem.

\end{abstract}

\thispagestyle{empty}

\newpage
 \thispagestyle{empty}
 \setcounter{tocdepth}{2}
 
      \tableofcontents

 \thispagestyle{empty}
 \newpage

 \thispagestyle{empty}
 \newpage 
 \setcounter{page}{1}
\section{Introduction}

A major goal of quantum computing research is to understand the nature of problems that admit superpolynomial speedups over classical computation. Prototypical examples of problems admitting such dramatic quantum advantage---Simon's problem~\cite{Sim97}, {\sc Period-Finding}~\cite{Sho97}, {\sc Forrelation}~\cite{Aar10,AA18}, among others---are all highly structured. In Simon's problem and {\sc Period-Finding},  structure comes in the form of a hidden subgroup; in {\sc Forrelation},  structure comes in the form of a promise that one function is correlated with the Fourier transform of another. 

As Aaronson writes in~\cite{Aar22}, three decades of quantum algorithms research suggest a ``law of conservation of weirdness": Superpolynomial quantum speedups appear limited to settings where quantum algorithms can exploit some form of global structure to concentrate amplitudes on the correct answers.  Understanding the extent to which such a law holds is central to understanding the power and limitations of quantum computing. For example, it would further explain why superpolynomial speedups have been prevalent in number-theoretic cryptography but not for $\mathsf{NP}$-complete problems.

A particularly elegant formalization of such a law, in the query model where most quantum algorithms operate, is given by the following conjecture:

\newtheorem*{conjecture*}{Conjecture}

\begin{conjecture*}[The simulation conjecture]
Let $\mcA$ be a $t$-query quantum algorithm. For every $\eps,\delta > 0$, there is a classical algorithm that makes $\poly(t,1/\eps,1/\delta)$ queries and approximates $\mcA$'s acceptance probabilities to within an additive $\eps$ on $1-\delta$ fraction of inputs. 
\end{conjecture*}

This conjecture first appeared in~\cite{Aar05,Aar08,AA14} and has been attributed to folklore dating back to 1999.  To see how it relates to the ``need for structure", note that it implies that $F : \zo^N \to \{0,1,*\}$ can admit a superpolynomial quantum-classical query separation only if $F^{-1}(\{0,1\})$ is a very small set, corresponding to a highly specific promise on $F$'s inputs.\footnote{In more detail, if $|S|\ge \delta 2^N$ where $S\coloneqq F^{-1}(\{0,1\})$, the simulation conjecture implies the existence of a classical algorithm that makes $\poly(t,1/\delta,1/\alpha)$ queries and computes $F$ on  $1-\alpha$ fraction of $S$. For total functions ($S = \zo^N$), it has long been known that their quantum and classical query complexities are polynomially related~\cite{BBCMdW01}.} Relatedly, this conjecture, if true, would rule out the standard query-complexity approach to proving $\mathsf{BPP} \ne \mathsf{BQP}$ relative to a random oracle---a similarly  longstanding  open problem~\cite{FR99}. Both the simulation conjecture and the status of $\mathsf{BPP}$ vs.~$\mathsf{BQP}$ relative to a random oracle appear in Aaronson's list of ``Ten semi-grand challenges for quantum computing theory"~\cite{Aar05} and Fortnow's list of ``Open oracle  questions for the 21st century"~\cite{For21}.

\paragraph{The Aaronson--Ambainis Conjecture.} In~\cite{Aar08,AA14}, Aaronson and Ambainis reduced the simulation conjecture to a basic question about low-degree polynomials:

\begin{conjecture*}[Bounded low-degree polynomials have influential variables]
Let $p : \zo^N \to [0,1]$ be a bounded degree-$t$ polynomial. There is a variable $i\in [N]$ such that 
\begin{equation*} \Inf_i(p) %
\ge \poly\left(\frac{\Var(p)}{t}\right).
\end{equation*}
\end{conjecture*}

This Aaronson--Ambainis conjecture has since become the predominant approach to tackling the simulation conjecture. It is appealing because it is a natural, ``quantum-free" statement about polynomials, making it possible for the powerful toolkit of discrete Fourier analysis to be brought to bear on the simulation conjecture. %
 Indeed, \cite{AA14} noted that a weaker lower bound with $\exp(t)$ in place of $\poly(t)$ already follows from a Fourier-analytic result of Dinur, Friedgut, Kindler, and O'Donnell~\cite{DFKO07}, and this in turn implies a weaker version of the simulation conjecture with an $\exp(t)$ dependence. 

Despite almost two decades of effort, this remains the best bound for the Aaronson--Ambainis conjecture---and the simulation conjecture. The same bound has been rederived a couple of times~\cite{OZ16,DMP19}, but all known proofs hit a technical barrier at $\exp(t)$. (Briefly, this is because they all rely crucially on the hypercontractive inequality~\cite{Bon70}, a useful technique for analyzing low-degree polynomials, but one for which an exponential dependence on degree is inherent~\cite{ODBook}.) Other works on the conjecture include~\cite{Mon12,BB14,FHKL16,BSdW22,LZ23,Gut24,Bha25}.

\section{This work}

While the Aaronson--Ambainis conjecture has established itself as a problem of independent interest and import in discrete Fourier analysis~\cite{ODo12,Sim14}, it is natural to ask if we have made the simulation conjecture substantially more difficult by viewing quantum algorithms abstractly as polynomials. Indeed, Aaronson and Ambainis themselves concluded their paper with the suggestion that one could conceivably resolve the simulation conjecture without going through their conjecture, by reasoning directly about quantum algorithms. 

\paragraph{Query weights.} In this work we propose and explore one such approach, based on one of the earliest lower bound techniques in quantum computing. Introduced in the seminal work of Bennett, Bernstein, Brassard, and Vazirani~\cite{BBBV97}, the {\sl query weights} of a $t$-query quantum algorithm $\mcA$ are a way of tracking how it allocates its query budget among the $N$ variables. For each input $x \in \{0,1\}^N$, query $t' \in [t]$, and variable $i\in [N]$, let $\lvert \psi^{(t')}(x)\rangle$ denote the state immediately before the $t'$-th query and $\Pi_i$ denote the projector onto the subspace where the query register equals $i$. We define the quantities:    
\[
w_i^{(t')}(x) := \big\langle\psi^{(t')}(x)\big\lvert \Pi_i \big\rvert \psi^{(t')}(x)\big\rangle,
\qquad
W_i(x) := \sum_{t'=1}^t w_i^{(t')}(x),
\]
and note that $\sum_{i=1}^N w_i^{(t')}(x) = 1$ whereas $\sum_{i=1}^N W_i(x) = t$. 
Roughly speaking, $w_i^{(t')}(x)$ is the quantum analogue of the probability that $\mcA$'s $t'$-th query is to the $i$-th coordinate, and $(W_1(x),\ldots,W_N(x))$ is a summary of how $\mcA$ allocates its $t$ queries among the $N$ coordinates. While these quantities are defined for a specific input $x$, we will be mostly concerned with the aggregate quantities $\E[w_i^{(t')}(\bx)]$ and $\E[W_i(\bx)]$ for a uniform random $\bx\sim\zo^N$. 

\subsection{A new conjecture}  Now suppose $\mcA$ is an algorithm for a function $F : \zo^N \to \{0,1,*\}$. We can assume that neither $\Pr[F(\bx)=1]$ nor $\Pr[F(\bx)=0]$ is too small, since otherwise $F$ is well-approximated by a constant function.  We conjecture that there must then be a ``heavy variable", one for which $\mcA$ allocates a lot of its query weight:

\medskip 

 \begin{tcolorbox}[colback = white,arc=1mm, boxrule=0.25mm]

\begin{conjecture}[Query-efficient quantum algorithms have heavy variables]
\label{conj:our conjecture}
    Let $\mcA$ be a $t$-query quantum algorithm and $f : \zo^N \to [0,1]$ denote its acceptance probability. There is a variable $i \in [N]$ such that
    \begin{equation*}
         \Ex_{\bx\sim\zo^N}[W_i(\bx)] \ge \poly\paren*{\frac{\delta}{t}} \quad\quad\text{ where }\delta \coloneqq \min\big\{{\Pr[f(\bx) \ge \lfrac{2}{3}]}, \Pr[f(\bx) \le \lfrac{1}{3}]\big\}.
    \end{equation*}%
\end{conjecture}

\end{tcolorbox}
\medskip

For simplicity, suppose $\delta = 0.1$ and so the conjectured lower bound is $\E[W_i(\bx)] \ge \poly(1/t)$. Since the average coordinate has weight $t/N$, this $i$-th coordinate can be seen to receive a disproportionately large amount of weight. \Cref{conj:our conjecture} therefore formalizes a sense in which query-efficient quantum algorithms for balanced problems cannot ``evenly distribute their query budget" among all~$N$ coordinates.

We show that the Aaronson--Ambainis conjecture implies~\Cref{conj:our conjecture} (\Cref{claim:AA-to-us}). This follows from the fact that query weights upper bound influences after suitable normalization. On the other hand, we will also show that~\Cref{conj:our conjecture} is already sufficient to yield the simulation conjecture (\Cref{claim:our conjecture implies simulation conjecture}). %

\begin{remark}[Strong versions of~\Cref{conj:our conjecture} and the simulation conjecture]%
        We will also consider a strong version of~\Cref{conj:our conjecture} where the $\poly(\delta)$ dependence is replaced with a $1/\polylog(1/\delta)$ dependence (see~\Cref{conj:our conjecture strong}). This version implies a strong version of the simulation conjecture with a $\polylog(1/\delta)$ dependence.  As we discuss in~\Cref{sec:implications for random oracle separations}, these strong versions carry added implications for random-oracle separations.%
\end{remark}

\subsubsection{Comparison with the Aaronson--Ambainis conjecture}

\paragraph{Query weights vs.~influences.} Related to the fact that~\Cref{conj:our conjecture} is specific to quantum algorithms is that it concerns query weights rather than influences. Both are measures of a variable's ``importance". However, query weights are a property of a quantum algorithm $\mcA$, whereas influences are a property of its acceptance probability $f : \zo^N\to [0,1]$. Two algorithms with different query weights can share the same acceptance probability and the same influences. Through~\Cref{conj:our conjecture}, we are therefore proposing a ``syntactic"/``whitebox" approach to the simulation conjecture, which contrasts with Aaronson and Ambainis's ``semantic"/``blackbox" approach. 

While both types of approaches have their advantages, our work highlights two advantages of query weights. The first, as already mentioned, is that~\Cref{conj:our conjecture} can only be easier to prove than the Aaronson--Ambainis conjecture. As we will soon discuss in~\Cref{sec:intro-parallel}, this enables us to make progress on the simulation conjecture. The second is that query weights are, in a formal sense, algorithmically much easier to estimate than influences. As we will discuss in~\Cref{sec:implications for random oracle separations}, we use this to derive new implications for the status of $\BPP$ vs.~$\BQP$ relative to a random oracle.

\paragraph{The hybrid method vs.~the polynomial method.} As a related distinction, Aaronson and Ambainis's approach to the simulation conjecture is based on the polynomial method for quantum query lower bounds~\cite{BBCMdW01}. Ours, as we detail in~\Cref{sec:technical overview}, is instead based on the hybrid method~\cite{BBBV97}, which uses query weights to quantify the progress a quantum algorithm is able to make towards distinguishing nearby inputs. In their paper,~\cite{BBBV97} used this method to lower bound the query complexity of unstructured search (and notably, established the optimality of Grover's algorithm even before its  discovery). The hybrid method can be viewed as a precursor to the adversary method~\cite{Amb02}, which has in turn grown in sophistication and power over the years, culminating in proofs that the general adversary bound characterizes quantum query complexity~\cite{HLS07,Rei09,LMRSS11}. Our work shows that the basic hybrid method is already useful for the simulation conjecture. An avenue for future work is to see if more can be gained by combining our techniques with the full generality of the adversary method.

\subsection{Parallel quantum algorithms} \label{sec:intro-parallel}

As evidence of the viability of our approach, we make progress on~\Cref{conj:our conjecture}, and hence the simulation conjecture, by taking the {\sl parallelism} of quantum algorithms into account. While a sequential query algorithm makes a single query after another, a parallel algorithm is able to make multiple queries at once:

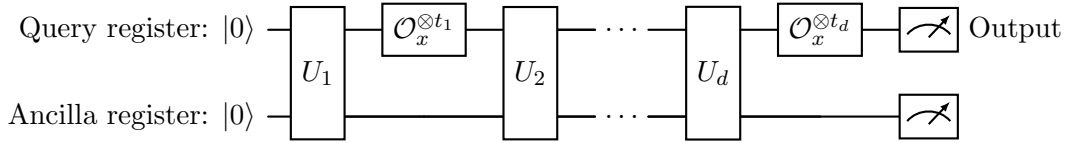
\begin{figure}[h]
  \centering
  \begin{quantikz}
    \lstick{Query register: $\ket{0}$} &[-2mm] \gate[wires=2]{U_1}  & \gate{\mcO_x^{\otimes t_1}} & \gate[wires=2]{U_2} & \ \ldots\ \qw & \gate[wires=2]{U_d} & \gate{\mcO_x^{\otimes t_d}}  & \meter{}  \rstick{\text{Output}} \\
    \lstick{Ancilla register: $\ket{0}$} & & \qw & \qw & \ \ldots\ \qw & \qw & \qw & \meter{}
  \end{quantikz}
    \captionsetup{width=.9\linewidth}

\caption{A $t$-query $d$-round quantum algorithm. Each round $r\in [d]$ consists of a unitary operator $U_r$ followed by $t_r$ parallel queries. The total number of queries across all rounds satisfies $\sum t_r = t$.}
  \label{fig:RTQuantumAlgo intro}
\end{figure}

Round complexity is the query model's abstraction of  {\sl circuit depth}, while the number of parallel queries per round is analogous to circuit width. Round complexity is also synonymous with {\sl adaptivity}, with $1$-round algorithms being most commonly called nonadaptive.

In quantum computing, as in classical computing, parallelism is a standard way to trade hardware for time. Small-depth quantum circuits are especially desirable because they also reduce exposure to decoherence, making them valuable on NISQ devices where full error correction is not yet available~\cite{Pre18}. %
Their importance persists even in the fault-tolerant setting, since  error correction and fault-tolerant implementations incur substantial overheads in practice~\cite{Got10,Aar22}. For these reasons, parallel quantum algorithms have received significant attention in the query model~\cite{vD98,Zal99,GR04,NY04,Mon10,MJM15,KLPY10,JMdW17,Bur19,GSTW24,CGV25}.

We defer the statement of our result regarding~\Cref{conj:our conjecture} to the body of the paper and state only its implication for the simulation conjecture: 

    \medskip 

 \begin{tcolorbox}[colback = white,arc=1mm, boxrule=0.25mm]

\begin{theorem}
\label{thm:d squared intro}
    Let $\mathcal{A}$ be a $t$-query $d$-round quantum algorithm. For every $\eps,\delta > 0$, there is a classical algorithm that makes 
    \[ T \coloneqq 2^{O(d^2)}\cdot (t\log(1/\delta)/\eps)^{O(d)}\]
    queries and  approximates $\mathcal{A}$'s acceptance probabilities to within an additive $\eps$ on $1-\delta$ fraction of inputs. 
\end{theorem}
\end{tcolorbox}
\medskip

\Cref{thm:d squared intro} confirms (the strong version of) the simulation conjecture for constant-round quantum query algorithms, showing that such algorithms cannot achieve a superpolynomial separation for unstructured problems. 

Our result continues to carry consequences in the regime of superconstant rounds. Note that $d\le t$ always, and our parameters improve upon the $T \le \exp(t)$ bound implied by~\cite{DFKO07} as long as $d\le O(\sqrt{t})$. In particular,~\Cref{thm:d squared intro} also shows that an  {\sl exponential} separation ($T \ge \exp(t^{\Omega(1)})$) for unstructured problems can only be achieved by quantum algorithms with {\sl polynomially} many rounds ($d \ge t^{\Omega(1)})$. This poses a dilemma for practical quantum advantage on unstructured problems. Exponential separations are the most likely to survive the real-world overheads of quantum error correction and fault tolerance, but these same overheads make physical implementations of polynomial-depth quantum circuits infeasible for the foreseeable future.

\paragraph{Nonadaptive algorithms.} Previously, the simulation conjecture was open even for $1$-round  algorithms. In fact, prior works point to such algorithms being surprisingly powerful in the context of quantum-classical query separations. Simon's problem, {\sc Period-Finding}, and {\sc Forrelation} are all solved by nonadaptive quantum query algorithms. Similarly, Yamakawa and Zhandry's recent exponential separation for an unstructured {\sl search} problem is also achieved by a nonadaptive quantum query algorithm~\cite{YZ24}. 

In light of~\Cref{thm:d squared intro}, any counterexample to the simulation conjecture must instead be based on quantum query algorithms that are highly adaptive.

\subsection{Implications for random-oracle separations} 
\label{sec:implications for random oracle separations}

Bernstein and Vazirani~\cite{BV97}, in their paper that defined $\mathsf{BQP}$, exhibited an oracle relative to which $\mathsf{BPP} \ne \mathsf{BQP}$. Shortly after, Fortnow and Rogers~\cite{FR99} asked if the same is true relative to a {\sl random} oracle, a problem that remains open today.  The simulation conjecture, if true, would rule out the standard approach of obtaining such a separation via query complexity.

We show that~\Cref{conj:our conjecture}, if true, would represent an even more significant ``technique-independent" barrier. Under its strong version, we get: %

    \medskip

 \begin{tcolorbox}[colback = white,arc=1mm, boxrule=0.25mm]
\begin{restatable}{theorem}{randomoracleunrelativizedequiv}
\label{thm:random oracle unrelativized equivalence}
    Assuming the strong version of~\Cref{conj:our conjecture},  $\mathsf{PromiseBPP}^{\bmcO} \ne \PromiseBQP^{\bmcO}$ for a random oracle $\bmcO$ if and only if $\mathsf{PromiseBPP} \ne \PromiseBQP$. 
\end{restatable}
\end{tcolorbox}

\medskip 

In words, the status of $\mathsf{PromiseBPP}$ vs.~$\PromiseBQP$ is {\sl equivalent} in the random oracle and unrelativized worlds. We prove the forward implication; the reverse implication is standard and holds unconditionally~\cite{BG81,FR99}. 
Under the standard version of~\Cref{conj:our conjecture}, we still get a barrier but lose the equivalence: If $\mathsf{PromiseBQP}^{\bmcO} \not\sse \mathsf{PromiseHeurBPP}^{\bmcO}$ then $\mathsf{PromiseBPP}\ne \PromiseBQP$ (see \Cref{thm:implication-weak-conjecture}).%

\paragraph{Unconditional equivalences.} By our proof of~\Cref{thm:d squared intro}, we further obtain {\sl unconditional} versions of~\Cref{thm:random oracle unrelativized equivalence} for parallel quantum circuits. For example, for the class of polylogarithmic-depth circuits ($\mathsf{QNC}$), we get:

    \medskip 

 \begin{tcolorbox}[colback = white,arc=1mm, boxrule=0.25mm]
\begin{theorem}%
    \label{thm:QNC equivalence}
    $\mathsf{PromiseQNC}^{\bmcO} \not\sse \mathsf{PromiseQuasiBPP}^{\bmcO}$ for a random oracle $\bmcO$ if and only if $\mathsf{PromiseQNC}\not\sse \mathsf{PromiseQuasiBPP}$. 
\end{theorem}
\end{tcolorbox}
\medskip 

Here $\mathsf{QuasiBPP}$ refers to randomized quasipolynomial time. 

$\mathsf{QNC}$, in addition to its practical relevance, has been shown to be a rather powerful subclass of $\mathsf{BQP}$. Notably, Cleve and Watrous~\cite{CW00} constructed $\mathsf{QNC}$ circuits for the quantum Fourier transform, thereby showing that Shor's factoring algorithm can be implemented using such circuits with polynomial-time classical pre- and post-processing.
Results like this suggest that $\mathsf{QNC}$ already captures certain important aspects of quantum advantage. Alas,~\Cref{thm:QNC equivalence} shows that providing evidence of such advantage in the form of a random-oracle separation is no easier than simply proving an unrelativized separation.  

To our knowledge,~\Cref{thm:random oracle unrelativized equivalence,thm:QNC equivalence} give the first natural examples of unresolved complexity-theoretic statements that can be shown to be  equivalent in the random-oracle and unrelativized worlds. %
While the sweeping Random Oracle Hypothesis~\cite{BG81} was quickly shown to be false~\cite{Kur82,Har85,CCGHHRR94}, it remains of interest to understand the extent to which the random-oracle world reflects the unrelativized world. Indeed,~\cite{CCGHHRR94} concluded by asking: ``It would be very interesting to know if there are identifiable problem classes for which the random oracle results do point in the right direction." Our results represent progress towards such an understanding.

\begin{remark}[Promise problems]
Our proofs of these equivalences require the generality of promise problems. In particular, we rely on the existence of a $\mathsf{PromiseBQP}$-complete problem, whereas $\BQP$ has no known complete problems.%\footnote{In~\cite{Gol06}, Goldreich makes the case that promise problems are just as fundamental to the development of complexity theory as total problems; see also~\cite{Aar14} for quantum-specific motivation. The survey~\cite{Wat08} even drops the $\mathsf{Promise}$-prefix and simply defines all standard classical and quantum complexity classes in terms of promise problems, writing ``little is lost and much is gained in shifting one’s focus to promise problems."  
%}  %
\end{remark}

\section{Independent work}

Independently and concurrently, Liu and Mutreja~\cite{LM26} gave a $\poly(t)$-query classical simulation of nonadaptive $t$-query quantum algorithms on most inputs. Their proof extends to handle quantum algorithms that make $O(\log t/\log\log t)$ adaptive singleton queries before one parallel batch of $t$  queries. They obtained these results by proving a ``dense indistinguishability conjecture" for such algorithms, and showing that this conjecture implies the simulation conjecture. 

Also independently and concurrently, Escudero Guti\'errez, Palazuelos, and Saucedo~\cite{EGPS26} gave a $O(tr^2)$-query classical simulation of nonadaptive $t$-query rank-$r$ quantum algorithms on most inputs. Rank-$r$ algorithms measure on all qubits and accept on at most $r$ of the measurement outcomes. 

After the submission of our manuscript and exchange of theorem statements, Liu and Mutreja~\cite{LM26} obtained a $t^{O(4^d)}$-query classical simulation of $t$-parallel $d$-round quantum algorithms on most inputs.

\section{Technical overview}
\label{sec:technical overview}

\paragraph{Regularity.} Quantum algorithms with low query weights will be at the heart of our analysis, and so it will be convenient for us to have the following terminology: 

    \begin{definition}[$\eta$-regularity]
    \label{def:regular}
    A quantum algorithm is {\sl $\eta$-regular} if its query weights $W_1(x),\ldots,W_N(x)$ satisfy
    \begin{equation*}
        \Ex_{\bx}[W_i(\bx)] \leq \eta \quad\quad\text{for all }i \in [N].
    \end{equation*}
\end{definition}

We will prove a ``regularity lemma" for this notion of regularity (\Cref{lem:regularity lemma}). We show that for any $t$-query quantum algorithm $\mcA$, there is a classical decision tree $T$ of depth $\poly(t,1/\eta,\log(1/\delta))$ such that $\mcA_\pi$ is $\eta$-regular for all but a $\delta$ fraction of its root-to-leaf paths $\pi$. This lemma is why the simulation conjecture follows from~\Cref{conj:our conjecture}: The latter, stated in the contrapositive, says that sufficiently regular quantum algorithms are well-approximated by constant functions, so $T$ itself serves as a good classical approximation of $\mcA$. %

\paragraph{\{$t$-query $d$-round\}  $\sse$ \{$t$-parallel $d$-round\}.} Instead of $t$-query $d$-round algorithms, all our results will hold for {\sl $t$-parallel} $d$-round algorithms. These algorithms can make $t$ parallel queries in {\sl each} round for $d$ rounds rather than $t$ queries in total (see~\Cref{def:RroundTparallelQalg}). Since our results concern the classical simulation of such algorithms, the same results also apply to $t$-query ones with the same parameters.

\subsection{Nonadaptive algorithms: The case of $d=1$}%
\label{subsec:overview-nonadaptive}
\violet{To demonstrate the power of reasoning through query weights, we begin by showing that~\Cref{conj:our conjecture}, and therefore the simulation conjecture, has a one-page proof in the case of nonadaptive quantum query algorithms.}

\violet{Here and throughout the paper, for a function $f : \zo^N \to [0,1]$ we will be interested in the sets: 
\[ \Acc_f\coloneqq\{x:f(x)\geq 2/3\} \quad\text{and}\quad  \Rej_f\coloneqq\{x:f(x)\leq 1/3\}.%
\]
When $f$ is clear from context, we simply write $\Acc$ and $\Rej$. The query weights $W_i(x)$ of nonadaptive algorithms do not depend on the input $x$, so we will just write $W_i$.
}

\begin{lemma}[\Cref{conj:our conjecture} is true for nonadaptive algorithms]
    \label{lem:nonadaptive-overview}
    Let $\mcA$ be a nonadaptive $t$-query quantum algorithm with query weights $W_1, \ldots, W_N$ and $f:\zo^N \to [0,1]$ denote its acceptance probability. There is some $i \in [N]$ for which
    \begin{equation*}
         W_i \ge \Omega\paren*{\frac{1}{t \ln(1/\delta)}} \quad\quad\text{ where }\delta \coloneqq \violet{\min\big\{{\Pr[\Acc_f]}, \Pr[\Rej_f]\big\}.}
    \end{equation*}
\end{lemma}
Our proof proceeds by contrapositive: We show that if $\linf{W}$ is sufficiently small \violet{(i.e.~if $\mcA$ is sufficiently regular)}, then one of $\Acc$ or $\Rej$ must be tiny. The only information about the function $f$ our proof uses is a rather simple consequence of \cite{BBBV97}'s hybrid method. For any two inputs $x,y \in \zo^N$, define the distance between them as,
\begin{equation*}
    \dist_W(x,y) \coloneqq \sum_{i \in \Delta(x,y)} W_i \quad\quad \text{where}\quad \Delta(x,y) \coloneqq \set{i \in [N]: x_i \neq y_i}.
\end{equation*}
The hybrid method shows that $\dist_W(x,y) \geq \Omega(1)$ for any $x \in \Acc$ and $y \in \Rej$. By the triangle inequality, this means every input $x \in \zo^N$ is far from at least one of these sets: 
\begin{equation*}
    \dist_W(x,\Acc) + \dist_W(x,\Rej) \geq \Omega(1) \quad\quad\text{where}\quad \dist_W(x,A) \coloneqq \min_{y \in A}\dist_W(x,y).
\end{equation*}
Consequently, the average input is far from one of these two sets: 
\begin{equation*}
    \max\bigg\{\Ex_{\bx}\bracket*{\dist_W(\bx,\Acc)}, \Ex_{\bx}\bracket*{\dist_W(\bx,\Rej)} \bigg\} \geq \lfrac{1}{2} \cdot \Omega(1) = \Omega(1).
\end{equation*}
At this point, we apply McDiarmid's bounded differences inequality~\cite{McD89,DP09}: For any set \violet{$A\sse \zo^N$} and inputs $x,x^{\oplus i}$ that differ in only the $i^{\text{th}}$ coordinate we have $\abs*{\dist_W(x, A) - \dist_W(x^{\oplus i},A)} \leq W_i$. Therefore,
\begin{equation*}
     \Prx_{\bx}\big[\dist_W(\bx,A) \le \E[\dist_W(\bx,A)] - \eps\big] \le  \exp\left(-\frac{2\eps^2}{\| W\|_2^2}\right) \leq \exp\left(-\frac{2\eps^2}{\lone{W}  \linf{W}}\right), 
\end{equation*}
where the first inequality is by McDiarmid's and the second is H\"older's. A $t$-query quantum algorithm has total query weight $\lone{W} \leq t$, and so we get that: %
\begin{align*}
    \min\set*{\Pr[\Acc],\Pr[\Rej]} &= \min\set*{\Prx_{\bx}[\dist_W(\bx,\Acc) = 0], \Prx_{\bx}[\dist_W(\bx,\Rej) = 0]} \\ &\leq \exp\left(-\Omega\paren*{\frac{1}{t\cdot  \linf{W}}}\right).
\end{align*}
Rearranging completes the proof of~\Cref{lem:nonadaptive-overview}.

\subsection{Warm-up for $d \ge 2$: A doubly exponential bound}
\label{subsec:overview-warmup}
Extending these ideas to adaptive algorithms (i.e.~$d \ge 2$) requires overcoming a few challenges. We will first illustrate these ideas with a simpler proof that achieves quantitatively weaker parameters. This warm-up suffices to prove~\Cref{conj:our conjecture}, and therefore the simulation conjecture, when $d =O(1)$. However, its performance degrades quickly. For example, when $d = \violet{\polylog(t)}$, an important case corresponding to $\QNC$, it gives a bound that is exponential in $t$. %

\begin{lemma}[\violet{Parallel quantum algorithms have heavy variables; warm-up}] %
    \label{lem:warmup-overview}
    Let $\mcA$ be a $t$-parallel $d$-round quantum algorithm with query weights \violet{$W_1(x), \ldots, W_N(x)$} %
    and $f:\zo^N \to [0,1]$ denote its acceptance probability. There is some $i \in [N]$ for which
    \begin{equation*}
         \Ex_{\bx \sim \zo^N}[W_i(\bx)] \ge \paren*{dt \ln(1/\delta)}^{-\Omega(2^d)} \quad\quad\text{ where }\delta \coloneqq \violet{\min\big\{{\Pr[\Acc_f]}, \Pr[\Rej_f]\big\}}.
    \end{equation*}
\end{lemma}

\violet{We give a full proof of~\Cref{lem:warmup-overview} in~\Cref{sec:2^r} and describe the main ideas here.} We once again apply the hybrid method, but now need to be more careful. \violet{The obvious but key distinction is that when $d \ge 2$,} the queries may depend on the input. In this setting, the hybrid method gives
\[
    \dist_{W(x)}(x,y) \geq \Omega(1/d) \quad\quad\text{ for all } x\in \Acc\text{ and }y \in \Rej, \label{eq:basic hybrid}
\]
\violet{where crucially, the distance measure now depends on $x$. This introduces two difficulties.} The first is that our prior approach based on McDiarmid's inequality assumes a fixed vector $W$ satisfying that $\dist_W(x, A) - \dist_W(x^{\oplus i},A) \leq W_i$ for all inputs $x$. With this new input-dependent distance measure, we would instead want to bound $\dist_{W(x)}(x, A) - \dist_{W(x^{\oplus i})}(x^{\oplus i},A)$. The mismatch between the two distance measures makes this quantity more difficult to analyze.

\paragraph{Talagrand's convex-distance inequality.} To get around this first difficulty, we use a remarkable result due to Talagrand~\cite{Tal96}. \violet{The standard statement of this inequality is given in~\Cref{thm:Talagrand}; we use a straightforward corollary:}%
\begin{corollary}[Easy consequence of Talagrand's convex-distance inequality; see~\Cref{cor:talagrand-linf}]
    \label{cor:Talagrand-overview}
    For any set $A \subseteq \zo^N$, vectors $\set{W(x)}_{x \in \zo^N}$ all satisfying \violet{$\lone{W(x)} \leq t$}, and thresholds $\eps,\gamma \geq 0$,
    \begin{equation*}
        \min\set*{\Pr[A],\Prx_{\bx}\bracket*{\dist_{W(\bx)}(\bx, A) \geq \eps }} \leq \exp\paren*{-\Omega\paren*{\lfrac{\eps^2}{t\gamma}}}+ \Prx_{\bx}[\linf{W(\bx)} \geq \gamma].
    \end{equation*}
\end{corollary}

\violet{In words,~\Cref{cor:Talagrand-overview} says that if $A$ is a large set and most $W(x)$'s have small $\| \cdot \|_\infty$, then most $x$'s are close to $A$, even if distance is measured in an $x$-specific manner according to $W(x)$.}
\violet{Note that the $d=1$ case is an immediate consequence of \Cref{cor:Talagrand-overview}}: In this case there is a single $W$ with $\linf{W} \le \eta$ and $\lone{W} \le t$, where for every $x \in \Acc$, we have $\dist_W(x,\Rej) \geq \Omega(1)$. Applying~\Cref{cor:Talagrand-overview} with $A = \Rej$ shows that either $\Acc$ or $\Rej$ must be tiny. The real power of~\Cref{cor:Talagrand-overview} comes in the setting of $d\ge 2$ where $W$ may now depend on $x$.

\violet{However, in order to apply~\Cref{cor:Talagrand-overview} we have to then overcome a second difficulty, that of bounding $\Pr[\| W(\bx)\|_\infty \ge \gamma]$ given only the assumption that every individual coordinate of $W(\bx)$ is small in expectation (i.e.~the assumption of $\eta$-regularity).} One could, of course, apply Markov's inequality on each individual coordinate $W_i(\bx)$. After a union bound over the $N$ coordinates, this gives
\begin{equation*}
     \Prx_{\bx}\bracket[\big]{\linf{W(\bx)} \gg  N  \eta} \ll 1.
\end{equation*}
\violet{Unfortunately, this is too lossy of an approach as it requires $\eta \ll 1/N$. The depth of the classical decision tree from our regularity lemma scales with $1/\eta$, and so this would result in a trivial bound.}   %

\paragraph{Exponential tail bounds for query weights, proved inductively round-by-round.} Instead, we will show that $W_i(\bx)$ \violet{for any fixed $i\in [N]$} satisfies an exponential tail bound, \violet{which translates into only a logarithmic dependence on $N$ when we take a union bound}. Such a logarithmic dependence can then be removed using standard techniques \cite{DFKO07}. To obtain such a bound, we first break down the query weights into their individual \violet{round} contributions,
\begin{equation*}
    W_i(x) = w_i^{(1)} + w^{(2)}_i(x)+ \cdots + w^{(d)}_i(x).
\end{equation*}
\violet{We will bound $W_i(x)$ by bounding each $w_i^{(r)}(x)$.}
The first term $w_i^{(1)}$ does not depend on the input $x$ and is easy to handle. The key observation is that for any $r \geq 2$ and $i \in [N]$, the quantity $w^{(r)}_i(x)$ is itself the acceptance probability of an $(r-1)$-round algorithm. \violet{We can therefore use Talagrand's inequality to not only analyze the acceptance probabilities of quantum algorithms as sketched above, but also to analyze their query weights inductively:} %

\violet{
\begin{lemma}[Inductive step]
    \label{lem:query-weights-tail-overview} Consider any $t$-parallel $d$-round quantum algorithm. Fix $i\in [N]$ and $r\ge 2$ and suppose $\E_{\bx}[w_i^{(r)}(\bx)]\le \gamma_r$. Then for any $\gamma_{r-1}$, we have: 
    \begin{equation*}
        \Prx_{\bx}\big[w^{(r)}_i(\bx) \geq 4\gamma_r\big] \leq \Prx_{\bx}\bracket*{\max_{r' \le r-1}\set*{\linf{w^{(r')}(\bx)}} \geq \gamma_{r-1}} + \exp\paren*{-\Omega\paren*{\frac{\gamma_r^2}{\gamma_{r-1}r^4 t}}}.
    \end{equation*}
\end{lemma}
}

\violet{In words,~\Cref{lem:query-weights-tail-overview} reduces the task of controlling the round-$r$ query weights (LHS of above)  to that of controlling the round-$r'$ query weights for $r'\le r-1$ (RHS of above). Importantly, it translates an exponential tail bound on round-$r'$ weights into an exponential tail bound for round-$r$ weights, assuming only a bound on the expectation of round-$r$ weights.} %

To utilize \Cref{lem:query-weights-tail-overview}, we choose thresholds $\gamma_1 \le \cdots\le \gamma_d$ and define, for each round $r$: 
\begin{equation*}
    \Bad^{(r)} \coloneqq \set{x \in \zo^N: \linf{w^{(r)}(x)} > \gamma_r} \quad\quad\text{and} \quad\quad \Good \coloneqq \zo^N \setminus \bigcup_{r \in [d]} \Bad^{(r)}.
\end{equation*}
\Cref{lem:query-weights-tail-overview} allows us to conclude that $\Bad^{(r)}$ is small provided that:
\begin{enumerate}
    \item The choice of thresholds satisfies $\gamma_{r-1} \ll \gamma_r^2$, \violet{chosen so that the second term on the RHS of~\Cref{lem:query-weights-tail-overview} is sufficiently small.} %
    \item The prior sets $\Bad^{(1)}, \ldots, \Bad^{(r-1)}$ are all sufficiently small.
    \item The query weights are sufficiently small in expectation, \violet{which is implied by our regularity assumption. (As in the nonadaptive case, our proof of~\Cref{lem:warmup-overview} proceeds via the contrapositive, showing that sufficiently regular algorithms must be biased.)} %
\end{enumerate}
We set the parameters in such a way that all these conditions are satisfied for all rounds $r$, meaning that most $x's$ are in $\Good$. \violet{This yields a bound on $\Prx\big[\| W(\bx)\|_\infty\ge \gamma\big]$, which allows us to apply Talagrand's inequality yet again to prove \Cref{lem:warmup-overview}.} 

The requirement that $\gamma_{r-1} \ll \gamma_r^2$ is why \Cref{lem:warmup-overview} has a doubly exponential dependence on~$d$: Chaining this inequality for all $r \in [d]$ gives that $\gamma_1 \ll (\gamma_d)^{2^{d}}$.  \violet{We then have to set the regularity parameter~$\eta$ to be correspondingly small for the base case of our inductive proof to go through. }

\subsection{An improved bound for $d \ge 2$ via higher-order statistics}
\label{subsec:overview-improved}
Using a few additional ideas and a much more careful analysis we are able to prove the following.
\begin{theorem}[\violet{Parallel quantum algorithms have heavy variables}; improved bound]%
    \label{thm:rsquared-overview}
    Let $\mcA$ be a $t$-parallel $d$-round quantum algorithm with query weights \violet{$W_1(x), \ldots, W_N(x)$} and $f:\zo^N \to [0,1]$ denote its acceptance probability. There is some $i \in [N]$ for which
    \begin{equation*}
         \Ex_{\bx \sim \zo^N}[W_i(\bx)] \ge 2^{-\Omega(d^2)} \cdot (t\log(1/\delta))^{-\Omega(d)} 
         \quad\quad\text{ where }\delta \coloneqq \violet{\min\big\{{\Pr[\Acc_f]}, \Pr[\Rej_f]\big\}}.
    \end{equation*}
\end{theorem}
Our proof of \Cref{thm:rsquared-overview} will be based on more fine-grained information about the quantum algorithm than in \Cref{lem:warmup-overview}. Rather than just  analyzing the query weights of \textsl{individual coordinates}, we will analyze the query weights of \textsl{sets}. On an input $x\in \zo^N$, each round $r\in [d]$ of a $t$-parallel algorithm induces a distribution $\mcD_x^{(r)}$ over subsets of $[N]$ of size at most $t$.  %
These distributions encode more information than the query weights of coordinates, as the latter can be easily recovered using the simple equality,
\begin{equation*}
    w_i^{(r)}(x) = \Prx_{\bS \sim \mcD_x^{(r)}}[i \in \bS].
\end{equation*}
We can now describe the main differences between the proofs of \Cref{thm:rsquared-overview} and \Cref{lem:warmup-overview}.

\paragraph{$\Good$ vs $\Bad$.}  \violet{The proof of~\Cref{lem:warmup-overview} is based on a partition of inputs into $\Bad^{(r)}$ and $\Good^{(r)}\coloneqq \zo^N\setminus \Bad^{(r)}$, where the definitions of these sets are chosen to facilitate applications of Talagrand's inequality to the inputs in $\Good^{(r)}$. Specifically, we consider an input $x$ good if $\| w^{(r)}(x) \|_\infty \le \gamma_r$. This can be interpreted as saying that $\mathcal{D}_x^{(r)}$ is ``well-spread" in the sense that 
\[ \Prx_{\bS \sim \mathcal{D}_x^{(r)}}[i\in \bS]\le \gamma_r \quad \text{for all $i\in [N]$}.  \]
This notion of goodness is therefore based on ``singleton statistics" (i.e.~``$1$-wise statistics") of $\mathcal{D}_x^{(r)}$. }

Goodness in the proof of~\Cref{thm:rsquared-overview} will instead be based on a broader set of $m$-wise statistics of $\mathcal{D}_x^{(r)}$ where $m\ge 1$. We now consider an input $x$ good if $\mathcal{D}_x^{(r)}$ is well-spread in the sense that 
\begin{equation} \Prx_{\bS\sim \mathcal{D}_x^{(r)}}[\bS\cap T \ne\emptyset] \le \gamma_r \text{ for all subsets $T\sse [N]$ of size $m_r$.}\label{eq:well spread overview}
\end{equation}

\violet{We say that $\mcD_x^{(r)}$ is {\sl $(m_r,\gamma_r)$-spread} if it satisfies the above, and {\sl $(m_r,\gamma_r)$-hit} otherwise.}

\violet{Recall that the bottleneck in our proof of~\Cref{lem:warmup-overview} was the recurrence relationship $\gamma_{r-1}\ll \gamma_r^2$, which resulted in $\gamma_1$ being doubly exponential in $d$. With this new dual-parameter notion of spreadness, with an appropriately large choice of $m_r$ our recurrence relationship becomes $\gamma_{r-1}\le c\gamma_r$ for some absolute constant $c$, which results in $\gamma_1$ only being singly exponentially in $d$. 

}

\paragraph{Stronger distance guarantees via a more careful hybrid method.} \violet{Note, however, that $(m_r,\gamma_r)$-spreadness becomes a stricter guarantee the larger $m_r$ is. Since the size of $\Good^{(r)}$ shrinks as $m_r$ becomes large, it is more challenging to bound the size of $\Bad^{(r)}$ when we work with this new notion. 

We overcome this in two ways. We first show that the hybrid method actually gives a stronger guarantee than what we used in the proof of~\Cref{lem:warmup-overview}. Our proof of~\Cref{lem:warmup-overview} uses the fact that if $x\in \Acc$ and $y\in \Rej$, then} 
\violet{
\begin{equation*}
    \sum_{r \in [d]} \dist_{w^{(r)}(x)}(x,y) \geq \Omega(1/d) \quad\text{where} \quad \dist_{w^{(r)}(x)}(x,y) \coloneqq \sum_{i\in \Delta(x,y)} w^{(r)}_i(x). %
\end{equation*}
Working directly with the distributions $\mcD_x^{(r)}$, the hybrid method gives the stronger bound: %
\begin{equation}
    \sum_{r \in [d]} \dist_{\mcD_x^{(r)}}(x,y) \geq \Omega(1/d) \quad\text{where} \quad \dist_{\mathcal{D}_x^{(r)}}(x,y) \coloneqq \Prx_{\bS\sim\mathcal{D}_x^{(r)}}[x_{\bS} \ne y_{\bS}]. \label{eq:new distance overview} %
\end{equation}
This bound is indeed stronger since 
\begin{align*}
\Prx_{\bS\sim\mcD_x^{(r)}}[x_{\bS}\ne y_{\bS}] &= \Prx_{\bS\sim\mathcal{D}_x^{(r)}}[\bS \cap \Delta(x,y)\ne \emptyset] \\
&\le \Ex_{\bS \sim \mathcal{D}_x^{(r)}}[|\bS \cap \Delta(x,y)|]\ =\ \sum_{i\in \Delta(x,y)} \Prx_{\bS\sim\mathcal{D}_x^{(r)}}[i\in \bS] \ = \ \sum_{i\in \Delta(x,y)} w_i^{(r)}(x).  
\end{align*} }

\paragraph{Our main concentration inequality.} \violet{Next, we prove a generalization of~\Cref{cor:Talagrand-overview}, the reformulation of Talagrand's inequality that we used repeatedly in the proof of~\Cref{lem:warmup-overview}. This version incorporates our new notions of well-spreadness (\ref{eq:well spread overview}) and distance (\ref{eq:new distance overview}):}

\begin{lemma}[Well-spread distributions cannot separate large sets; see~\Cref{lem:dist-to-hit}]
    \label{lem:dist-to-hit-overview}
   For any set $A \subseteq \zo^N$, distributions  $\set{\mcD_x}_{x \in \zo^N}$ supported on subsets of $[N]$ of size at most $t$, and thresholds $\eps,\gamma \geq 0$,
        \begin{equation*}
        \min\set*{\Pr[A], \Prx_{\bx}\big[ \dist_{\mcD_{\bx}}(\bx, A) \geq \eps\big] }\leq \exp\paren*{-\Omega\paren*{\frac{(\eps-\gamma)^2m}{\gamma t}}} + \Prx_{\bx}\big[\mcD_{\bx}\text{ is }(m,\gamma)\text{-hit}\,\big].
    \end{equation*} 
\end{lemma}
\violet{\Cref{lem:dist-to-hit-overview} recovers~\Cref{cor:Talagrand-overview} when $m=1$. We note the tension between the two terms on the RHS: As $m$ grows, the first becomes smaller whereas the second becomes larger.}

\violet{With~\Cref{lem:dist-to-hit-overview} in hand, our proof of~\Cref{thm:rsquared-overview} shares the same overall structure as that of~\Cref{lem:warmup-overview}. For each round $r\in [d]$, we consider a set of bad inputs $\Bad^{(r)} \sse \zo^N$, now defined as those whose distributions $\mathcal{D}^{(r)}_x$ are $(m_r,\gamma_r)$-hit. We argue inductively, using~\Cref{lem:dist-to-hit-overview}, that if $\Bad^{(1)},\ldots,\Bad^{(r-1)}$ are small then $\Bad^{(r)}$ is small as well. Having shown that most inputs fall in $\Good \coloneqq \zo^N \setminus \bigcup_r \Bad^{(r)}$, we then apply~\Cref{lem:dist-to-hit-overview} once more to conclude that $\Pr[\Acc]\cdot \Pr[\Rej\cap\Good]$ is small. 

}

\subsection{Implications for random-oracle separations: Proofs of~\Cref{thm:random oracle unrelativized equivalence,thm:QNC equivalence}}

\label{sec:implications for random oracle separations technical overview}

\violet{Separations in the query model yield relativized separations, but the converse does not necessarily hold: simulations in the query model do not necessarily yield relativized simulations. In particular, the simulation conjecture alone does not imply that $\mathsf{BPP}^{\bmcO} =  \mathsf{BQP}^{\bmcO}$ for a random oracle $\bmcO$. While it guarantees that every quantum Turing machine making $\poly(n)$ queries to an oracle has an approximator that only makes $\poly(n)$ classical queries, it does not guarantee that this approximator can be implemented using only $\poly(n)$ classical {\sl time steps}.

To prove~\Cref{thm:random oracle unrelativized equivalence}, we show that the specific approximators that we construct via~\Cref{conj:our conjecture} do admit time-efficient implementations assuming $\PromiseBPP= \PromiseBQP$. To sketch why, first recall that~\Cref{conj:our conjecture} implies the simulation conjecture via a regularity lemma, which we now state formally:}
\begin{restatable}[Regularity lemma]{lemma}{regularityoverview}
\label{lem:regularity lemma}
Let $\mcA$ be a $t$-query quantum algorithm. There is a classical
decision tree $T$ of depth $\poly(t,1/\eta,\log(1/\delta))$ such that
the restricted algorithm $\mcA_{\bpi}$ is $\eta$-regular  with probability at least 
$1-\delta$ with respect to a random root-to-leaf path $\bpi\sim T$ induced by a uniform random input $\bx\sim \zo^N$.
\end{restatable}

\violet{This tree $T$ is constructed using a simple greedy algorithm: If $\mcA$ is already $\eta$-regular, we are done. Otherwise, there must be at least one variable $i\in [N]$ with query weight larger than $\eta$. We query any such variable at the root of our decision tree and recurse on $\mcA_{x_i=0}$ and $\mcA_{x_i=1}$. We use a martingale argument to show that a random path through this tree will w.h.p.~terminate after not too many steps, and~\Cref{lem:regularity lemma} follows.}

\violet{For our proof of~\Cref{thm:random oracle unrelativized equivalence}, we further need an algorithmic version of this regularity lemma:}

\begin{lemma}[Algorithmic regularity lemma, see \Cref{cor:algorithmic_regularity_lemma} for the formal version]%
\label{lem:algorithmic regularity lemma overview}
\violet{If $\PromiseBPP = \PromiseBQP$, there is an efficient randomized algorithm that takes as input a representation of $\mcA$ as a quantum oracle circuit and on any input $\mcO\in\zo^N$ computes that path $\pi$ in $T$ that $\mcO$ follows.
}
\end{lemma}

\violet{A natural attempt at proving~\Cref{lem:algorithmic regularity lemma overview} would be to estimate the query weight of each variable $i\in [N]$, in search of one with weight at least $\eta$. This would be the root node in the path~$\pi$, and we can then recurse. However, such a search is too time consuming: A $\BQP$ oracle algorithm makes at most $\poly(n) = \polylog(N)$ queries to $\mcO\in \zo^N$, and so our classical simulation has to take at most $\polylog(N)$ classical time steps. We instead give a $\PromiseBQP$ algorithm for estimating {\sl sums} of query weights, $\sum_{i \in S} \Ex_{\bx}[W_i(\bx)]$ for any subset $S\sse [N]$.  With this subroutine in hand, a standard branch-and-prune algorithm (\'a la~\cite{GL89,KM93}) allows us to identify a variable of high query weight with only a logarithmic dependence on $N$.

\Cref{thm:random oracle unrelativized equivalence} then follows easily, %
and the proof of~\Cref{thm:QNC equivalence} is completely analogous.   
}

\violet{\paragraph{Complexity of finding high query weight vs.~high influence variables.} \cite{AA14}'s reduction of the simulation conjecture to their conjecture is also based on a regularity lemma,  where their notion of a regular quantum algorithm is one whose {\sl influences}, rather than query weights, are all low. They also showed that their regularity lemma can be made algorithmic, but under the much stronger assumption that $\mathsf{P} = \mathsf{P}^{\#\mathsf{P}}$ rather than $\mathsf{PromiseBPP} = \mathsf{PromiseBQP}$. The crux of this difference lies in the fact that our subroutine for finding a variable of high query weight,
as discussed above, is in $\PromiseBQP$, whereas~\cite{AA14}'s subroutine for finding one of high influence is in $\mathsf{NP}^{\#\mathsf{P}}$.

This highlights a dual advantage of query weights over influences. On one hand, since query weights upper bound influences (after suitable normalization, see~\Cref{prop:query-weight-to-inf}), any regularity lemma for low query weights also serves as one for low influences. And yet, somewhat paradoxically, our regularity lemma for query weights admits a much more efficient algorithmic implementation than~\cite{AA14}'s one for influences.

}

\section{Preliminaries}

\paragraph{Notation and naming conventions.} We write $[N]$ to denote the set $\{1,2,\ldots,N\}$. We use \textbf{boldface} letters, e.g.~$\bx,\by$, to denote random variables.  Unless otherwise stated, all probabilities are w.r.t.~uniform random $\bx\sim\zo^N$. For a set $A\subseteq \zo^N$, we write $\Pr[A]$ to denote $\Pr_{\bx}[\bx\in A]$. For an event $E$, we use $\Ind[E]$ for the $\zo$-valued indicator of whether $E$ occurs. For $x,y \in \zo^N$ we denote the symmetric difference of $x,y$ as
\begin{equation*}
    \Delta(x,y) = \set{i \in [N] : x_i \neq y_i}.
\end{equation*}
We use calligraphic font for several objects: oracles ($\mcO$), distributions $(\mcD)$, and algorithms $(\mcA)$. We use $\norm{U}$ for the operator norm of a matrix $U$, and $\lone{\ket{\psi}}, \ltwo{\ket{\psi}}, \linf{\ket{\psi}}$ for the standard $p$-norms.

Here and throughout the paper, for a function $f : \zo^N \to [0,1]$ (which we reserve to denote acceptance probabilities of a quantum algorithm) we define the sets:  
\[ \Acc_f\coloneqq\{x:f(x)\geq 2/3\} \quad\text{and}\quad  \Rej_f\coloneqq\{x:f(x)\leq 1/3\}. \]
When $f$ is clear from context, we simply write $\Acc$ and $\Rej$. 

\pparagraph{Distances.}
 For any $x, y \in \zo^N$ and $w \in \R_{\geq 0}^N$, we define
\begin{equation*}
    \dist_w(x,y) \coloneqq \sum_{i \in \Delta(x,y)} w_i = \sum_{i \in [N]} w_i \cdot \Ind[x_i \neq y_i].
\end{equation*}
Similarly, for distribution $\mcD$ over subsets of $[N]$,
\begin{equation*}
    \dist_{\mcD}(x,y) \coloneqq \Prx_{\bS \sim \mcD}[\bS \cap \Delta(x,y) \neq \emptyset].
\end{equation*}
Both of these measures extend in the natural way to a set $A \subseteq \zo^N$,
\begin{equation*}
    \dist_w(x,A) \coloneqq \min_{y \in A}\dist_w(x,y) \quad\text{and}\quad \dist_{\mcD}(x,A) \coloneqq \min_{y \in A} \dist_{\mcD}(x,y).
\end{equation*}

\pparagraph{Concentration inequalities.}
We will use the following concentration inequality due to Talagrand.
\begin{theorem}[Talagrand's convex-distance inequality, {\cite[Thm.~6.1 and Eq.~(6.5)]{Tal96}}]
    \label{thm:Talagrand}
    For any set $A \subseteq \zo^n$ and $t \geq 0$
    \begin{equation*}
        \Pr[A] \cdot \Pr\big[\dt(\bx,A) \geq t\big] \leq e^{-t^2/4},
    \end{equation*}
    where $\dt(\cdot,\cdot)$ is Talagrand's convex distance: %
    \begin{equation*}
        \dt(x,A) \coloneqq \mathop{\sup_{W \in \R^n_{\ge 0}}}_{\ltwo{W}\le 1} \dist_{W}(x,A). %
    \end{equation*}
\end{theorem}

We will also use the standard Azuma--Hoeffding bound (see e.g.~the textbook \cite{MUbook}).
\begin{theorem}[Azuma--Hoeffding inequality]
    \label{thm:azuma}
    Let $\bZ_0, \ldots, \bZ_D$ be a martingale satisfying $\abs*{\bZ_{i} - \bZ_{i-1}} \leq t$ almost surely for all $i \in [D]$. Then,
    \begin{equation*}
        \Pr[\bZ_D - \bZ_0 \geq \eps] \leq \exp\paren*{-\frac{\eps^2}{2Dt^2}}.
    \end{equation*}
\end{theorem}

\pparagraph{Restrictions.} A restriction of $\zo^N$ is denoted $\pi \in \{0,1,\star\}^N$. Its size, denoted $\abs*{\pi}$ is the number of $i \in [N]$ for which $\pi_i \neq \star$. %
For an input $x \in \zo^N$ we use $x_{\pi}$ to denote the restricted input,
\begin{equation*}
    (x_\pi)_i = \begin{cases}
       \pi_i&\text{if $\pi_i \neq \star$}\\
        x_{i}&\text{otherwise.}
    \end{cases}
\end{equation*}
We say $x$ is \textsl{consistent} with $\pi$ (or vice versa) if $x_{\pi} = x$. These restrictions can be appended to using the notation, for any $i \in [N], b \in \zo$
\begin{equation*}
    (\pi \cup \set{x_i = b})_j = \begin{cases}
        b&\text{if }j=i\\
        \pi_j&\text{otherwise.}
    \end{cases}
\end{equation*}
For a function $f:\zo^N \to [0,1]$ or algorithm $A:\zo^N \to \zo$ we use $f_{\pi}$ and $A_{\pi}$ to denote the function $x \mapsto f(x_{\pi})$ and algorithm $x \mapsto A(x_{\pi})$ respectively.

\pparagraph{Complexity classes.}
We will define many complexity classes of promise problems.
\begin{definition}[Promise problems]
    A promise problem $L \coloneqq  (L_{\YES}, L_{\NO})$ is specified by two subsets $L_{\YES}, L_{\NO} \subseteq \zo^{\star}$ satisfying $L_{\YES} \cap L_{\NO} = \emptyset$.
\end{definition}
\begin{definition}[Complexity classes]
    Let $L \coloneqq  (L_{\YES}, L_{\NO})$. Then,
    \begin{enumerate}
        \item $\PromiseBPP$ contains $L$ iff there is a randomized polynomial time Turing machine $A$ satisfying, for all inputs $x$,
        \begin{align*}
        x \in L_{\YES} &\implies \Pr_A[A(x) = 1] \geq 2/3,\\
        x \in L_{\NO} &\implies \Pr_A[A(x) = 1]\leq 1/3.
        \end{align*}
    \item $\PromiseBQP$ is identical to $\PromiseBPP$ except $A$ is allowed to be a quantum Turing machine.
    \item $\PromiseQuasi$ is identical to $\PromiseBPP$ except $A$ is allowed to run in time $n^{\polylog(n)}$.
    \item $\PromiseQNC$ contains $L$ iff there is an efficient classical Turing machine $A$ that on input length~$n$ writes down a polynomial-size polylogarithmic-depth quantum circuit $C_n$ such that, for all $n$ and inputs $x$ of length $n$,
    \begin{align*}
        x \in L_{\YES} &\implies \Pr[C_{n}(x) = 1] \geq 2/3,\\
        x \in L_{\NO} &\implies \Pr[C_{n}(x) = 1]\leq 1/3.
        \end{align*}
    \item $\PromiseHeurBPP$ contains $L$ iff for every constant $c$, there is a randomized polynomial time Turing machine $A$ satisfying, for all $n \in \N$,
    \begin{equation*}
        \Prx_{\bx \sim \Unif(\zo^n)}[A\text{ is incorrect on }\bx] \leq \frac{1}{n^c}.
    \end{equation*}
    where $A$ is incorrect on $x$ if $x \in L_{\YES}$ and $\Pr_A[A(x) = 1] < 2/3$ or $x \in L_{\NO}$ and $\Pr_A[A(x) = 1] > 1/3$.
    \end{enumerate}
\end{definition}

\subsection{Quantum query algorithms}
Let $x \in  \zo^N $ be an $N$-bit string $x_1\ldots x_N$. In the following, we consider the standard phase oracle model,
\begin{align*}
    \mcO_x:\ket{i} \ket{b} \mapsto (-1)^{x_ib}\ket{i} \ket{b} \qquad \text{for} \qquad i \in [N], b\in\zo
\end{align*}
which in general may act on superpositions.

For a quantum query algorithm we refer to both $\mcO_x$ and $x \in \zo^N$ as the \textit{oracle} and sometimes the algorithm's \textit{input}.
We say that a quantum algorithm makes a query to $x$ \textit{or} $\mcO_x$ whenever it applies $\mcO_x$. We will also say that a quantum algorithm makes a $t$-parallel query to $x$ \textit{or} $\mcO_x$ whenever it applies $\mcO_x^{\otimes t}$.

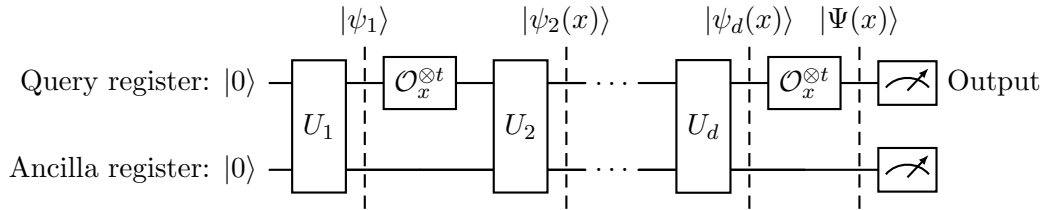
\begin{figure}[h]
  \centering
  \begin{quantikz}
    \lstick{Query register: $\ket{0}$} &[-2mm] \gate[wires=2]{U_1} \slice[style=black]{$\ket{\psi_1}$} & \gate{\mcO_x^{\otimes t}} & \gate[wires=2]{U_2} \slice[style=black]{$\ket{\psi_2(x)}$} & \ \ldots\ \qw & \gate[wires=2]{U_d} \slice[style=black]{$\ket{\psi_d(x)}$} & \gate{\mcO_x^{\otimes t}}  \slice[style=black]{$\ket{\Psi(x)}$} & \meter{}  \rstick{\text{Output}} \\
    \lstick{Ancilla register: $\ket{0}$} & & \qw & \qw & \ \ldots\ \qw & \qw & \qw & \meter{}
  \end{quantikz}
  \caption{A $t$-parallel $d$-round quantum query algorithm.}
  \label{fig:RTQuantumAlgo}
\end{figure}

\begin{definition}[$d$-round $t$-parallel quantum query algorithm (\Cref{fig:RTQuantumAlgo})]\label{def:RroundTparallelQalg}
    For any quantum oracle $\mcO_x$ acting on a Hilbert space $\mcH$, a $t$-parallel $d$-round quantum query algorithm is given by a sequence of unitaries
    $U_1,\ldots,U_d$, and a measurement operator $M_0$  
    acting on the tensor Hilbert space $\mcH^{\otimes t}$ and possibly ancillas. 
    For each $r \in [d]$, the intermediate state of the algorithm before the $r^{th}$ query is given by 
    \begin{align*}
        \psirx \coloneqq
U_{r}(\mcO_x^{\otimes t})\cdots U_2(\mcO_x^{\otimes t})U_1\ket{0}.
    \end{align*}
    The final state of the algorithm is $ \ket{\Psi(x)} \coloneqq (\mcO_x^{\otimes t})\psiRx$ and its acceptance probability is given by $\norm{M_0 \ket{\Psi(x)}}^2$ for a measurement operator $M_0$.
\end{definition}

\subsection{Query weights}
We use the hybrid method pioneered by \cite{BBBV97}. Roughly speaking, that method shows that if there are inputs $x,y$ for which the acceptance probability of $x$ is much more than $y$, then $x$ and $y$ must differ on indices with high ``query weight." In the case of sequential algorithms (i.e. $t=1$) in which the method has typically been applied, the query weight of index $i$ by a state $\ket{\psi}$ is defined as $w_i = \bra{\psi} |i\rangle\langle i| \otimes \Id \ket{\psi}$ where $ |i\rangle\langle i| \otimes \Id$ is the projector onto the $i$-th query index. 

We will use two natural generalizations of query weights for $t$-parallel algorithms. In the first, we allocate weight to an index $i$ if it was one of the $t$ indices queried. Let $\Pi_i^{(t)}\coloneqq \Id - \bigotimes_{j=1}^t (\Id - E_i^{(j,t)})$ be the projector onto the subspace where at least one of the $t$ queries is $i$, and where $E_i^{(j,t)}$ is the projector onto index $i$ in the $j$-th register.
\begin{definition}[(Index) query weight]\label{def:query_weight_t_parallel}
The query weight of an index $i \in [N]$ in round $r \in [d]$ on oracle input $x$ is defined as
\begin{equation*}
     w_i^{(r,t)}(x) \coloneqq \bpsirx \Pi^{(t)}_i \psirx.
\end{equation*}
The \textit{total (index) query weights} over $d$ rounds   is,
\begin{equation*}
    W_i^{(t)}(x) \coloneqq \sum_{r \in [d]} w_i^{(r,t)}(x). \label{eq:total_qweight}
\end{equation*}
\end{definition}
Note that the sum of the $t$-parallel query weights is bounded by $t$, since 
\begin{align}
\sum_{i \in [N]}\bra{\Psi}\Pi_i^{(t)}\ket{\Psi} \le \sum_{j=1}^t \sum_{i \in [N]}\bra{\Psi}E_i^{(j,t)}\ket{\Psi} = \sum_{j=1}^t \bra{\Psi}\Id\ket{\Psi} = t .\label{eq:t_parallel_index_query_weights_sum_to_t}
\end{align}
In the second generalization, we look at the query weight of an entire set $S \subseteq [N]$. Let $\Pi_S^{(t)}$ be the orthogonal projector onto the subspace where the set of distinct indices queried across all $t$ registers is exactly $S$, that is 
\begin{align*}
    \Pi_S^{(t)} = \sum_{\substack{(i_1, \dots, i_t) \in [N]^t \\ \{i_1, \dots, i_t\} = S}} \bigotimes_{j=1}^t \ket{i_j}\bra{i_j}.
\end{align*} 
\begin{definition}[(Set) query weight]\label{def:subset_query_weight}
The query weight of a set $S \subseteq [N]$ in round $r \in [d]$ on oracle input $x$ is defined as
\begin{equation*}
     w_S^{(r,t)}(x) \coloneqq \bpsirx \Pi^{(t)}_S \psirx.
\end{equation*}
The \textit{total (set) query weight} over all $d$ rounds is, 
\begin{equation*}
    W_S^{(t)}(x) \coloneqq \sum_{r \in [d]} w_S^{(r,t)}(x).\label{eq:total_subset_qweight} 
\end{equation*}
\end{definition}

In the following, we drop the superscript $t$ and write $w_i^{(r)}(x)$, $W_i(x)$, $w_S^{(r)}(x)$ and $W_S(x)$ when clear from context.

Since the subset projectors $\{\Pi_S\}_{S \subseteq [N]}$ are mutually orthogonal and cover the entire Hilbert space of the query registers, $\sum_{S \subseteq [N]} \Pi_S = \Id$.  Consequently, for any query state $\ket{\Psi}$, the sum of the exact subset query weights $w_S$ over all possible subsets $S$ is exactly $1$. Additionally, since there are at most $t$ distinct indices that can be queried at once, $\Pi_S = 0$ whenever $|S| > t$, and thus,
\begin{equation*}
    \sum_{S \subseteq [N] } w_S= \sum_{\substack{S \subseteq [N] \\ |S|\le t}} w_S = \sum_{\substack{S \subseteq [N] \\ |S|\le t}}  \bra{\Psi} \Pi_S \ket{\Psi}= \bra{\Psi} \sum_{\substack{S \subseteq [N] \\ |S|\le t}}  \Pi_S \ket{\Psi}= \bra{\Psi} \Id  \ket{\Psi} = 1.
\end{equation*}
It is thus natural to view $w_S$ as a distribution over sets $S \subseteq [N]$ of size $|S| \le t$. Throughout this paper we denote this distribution by $\mcD$ such that 
\begin{align*}
    \Prx_{\bS \sim \mcD_x^{(r)}}[\bS = S] = w_S^{(r)}(x).
\end{align*}

\subsection{Implications of the hybrid method}\label{sec:hybrid-imp}

We next turn to several corollaries of the hybrid method (\Cref{prop:hybrid_method}) that we use directly in the following sections.

First, we note that the query weights give rise to a natural notion of distance between quantum states which we will see bounds the progress of each query in the hybrid argument. For $x,y \in \zo^N$, the distance between $x$ and $y$ weighted by index query weights $\{w_i\}_{i \in [N]}$ is 
    \begin{align*}
        \dist_{w}(x,y) \coloneqq \sum_{i \in \Delta(x,y)} w_i
    \end{align*}
For $x,y \in \zo^N$, the distance between $x$ and $y$ given by distributions $\mcD$ induced by subset query weights $\set{w_S}_{S \subseteq [N]}$ is,
    \begin{align}
        \dist_{\mcD}(x,y) \coloneqq \Prx_{\bS \sim \mcD}[\bS \cap \Delta(x,y) \neq \emptyset] \label{def:weighted_dist_subset}
    \end{align}

The relationship between weighted distance by subset query weights and that of index query weights is that $\dist_{\mcD_x^{(r)}}(x,y)$ is a tighter measure of progress than $ \dist_{w^{(r)}(x)}(x,y) $ at each round $r \in [d]$:
\begin{align}
    \dist_{w^{(r)}(x)}(x,y) \geq \dist_{\mcD^{(r)}_x}(x,y) \qquad \text{and} \qquad
   \dist_{w^{(r)}(y)}(x,y) \geq \dist_{\mcD^{(r)}_y}(x,y)\label{eq:query_weights_vs_differentiating_set}.
\end{align}
This is due to the fact that any subset $S$ satisfying $S \cap \Delta(x,y) \neq \emptyset$ must contain some $i \in \Delta(x,y)$, and thus subset projectors (\Cref{def:subset_query_weight}) are majorized by the $t$-parallel index projectors (\Cref{def:query_weight_t_parallel}),
\begin{equation*}
    \sum_{i \in \Delta(x,y)} \Pi_{i} \succeq \sum_{\substack{S \,:\, S \cap \Delta(x,y) \neq \emptyset}} \Pi_{S}.
\end{equation*}

The hybrid method naturally imposes distinguishability constraints on both the final and intermediate quantum states of the algorithm. The hybrid method uses query weights as a tool for bounding these states. However, since the query weights are themselves derived from the intermediate states, we may equally regard them as \emph{representations} of those states---and as such, they must obey the same progress constraints. 
Concretely, for a \textit{binary test}  $\phi:\lset{S}_{S\subseteq [N]} \rightarrow \zo$, the sum of subset query weights indicated by the test, $\sum_{S \in \supp(\phi)} w_S^{(r)}(x) = \Ex_{\bS \sim \mcD^{(r)}_x}[\phi(\bS)]$ can be viewed as the output of an $r-1$-round quantum algorithm whose final measurement is $\sum_{S \in \supp(\phi)} \Pi_S$ and acceptance probability for each $x \in \zo^N$ given by,
\begin{equation*}
    \pnorm{2}{\sum_{S \in \supp(\phi)} \Pi_S \psirx}^2 = \sum_{S \in \supp(\phi)} w_S^{(r)}(x) = \Ex_{\bS \sim \mcD^{(r)}_x}[\phi(\bS)].
\end{equation*}

With this perspective, we ask whether the distributions $\mcD^{(r)}_x$ and $\mcD^{(r)}_y$ can be distinguished by a \textit{binary test} $\phi:\lset{S}_{S\subseteq [N]} \rightarrow \zo$. We say that $\phi$ distinguishes $\mcD^{(r)}_x$ from $\mcD^{(r)}_y$ with threshold $\gamma$ if
\begin{equation}
    \sum_{S \in \supp(\phi)} w_S^{(r)}(x) =\Ex_{\bS \sim \mcD^{(r)}_x}[\phi(\bS)] \geq \gamma \quad\quad\text{and}\quad\quad \sum_{S \in \supp(\phi)} w_S^{(r)}(y) =\Ex_{\bS \sim \mcD^{(r)}_y}[\phi(\bS)] \leq \gamma/2.
\end{equation}
In this language, we have the following corollary of the hybrid method (\Cref{prop:hybrid_method}).
\begin{corollary}\label{cor:quantum-model}
    Consider a $d$-round $t$-parallel query algorithm with intermediate pre-query states $\psix{r}$ for each $r \in [d]$ and final state $\ket{\Psi(x)}$. We use the notation from above to denote query weights $w_i^{(r)}(x)$, distributions $\mcD_x^{(r)}$ induced by subset query weights $w_S^{(r)}(x)$.  
    \begin{enumerate}[label=(\roman*)]
        \item \emph{Initial condition.} \label{item:inital_condition_cor}
        \begin{align*}
            w_i^{(1)}(x) =w_i^{(1)}(y)&&  \forall i \in [N], x,y \in \zo^N,\\
            w_S^{(1)}(x) =w_S^{(1)}(y) && \forall S \subseteq [N], x,y \in \zo^N 
        \end{align*}
        \item \emph{Progress evolution.} \label{item:progress_evolution_cor} For any $x\neq y \in \zo^N$, round $r \in [d]$, and threshold $\gamma$, if there is some test $\phi:\lset{S}_{S\subseteq [N]} \rightarrow \zo$ for which $\Ex_{\bS \sim \mcD^{(r)}_x}[\phi(\bS)] \geq \gamma$ and $ \Ex_{\bS \sim \mcD^{(r)}_y}[\phi(\bS)] \leq \gamma/2$ then
            \begin{equation}
            \label{eq:dist-bound-layer}
            \sum_{r' \in [r-1]}\sqrt{\dist_{w^{(r')}(x)}(x,y)} \ge \sum_{r' \in [r-1]}\sqrt{\dist_{\mcD_x^{(r')}}(x,y)} \geq \Omega(\sqrt{\gamma}).
        \end{equation} 
        In particular, applying this with $\phi(S) \coloneqq \Ind[i \in S]$, if $w_i^{(r)}(x) \geq \gamma$ and  $w_i^{(r)}(y) \leq \gamma/2$ then the same constraint holds.
        \item\label{item:distinguishibility_constraint_cor} \emph{Distinguishability constraint.} For any $x\neq y \in \zo^N$, if a measurement $M_0$ distinguishes the final states such that $f(x) = \pnorm{2}{M_0 \ket{\Psi(x)}}^2$ and $f(y) = \pnorm{2}{M_0 \ket{\Psi(y)}}^2$, then 
        \begin{equation}
             \ltup{2 \sum_{r \in [d]}\sqrt{\dist_{w^{(r)}(x)}(x,y)}}^2 \geq \ltup{2 \sum_{r \in [d]}\sqrt{\dist_{\mcD_x^{(r)}}(x,y) }}^2\ge \ltup{\sqrt{f(x)} - \sqrt{f(y)}}^2.\label{eq:final_condition_implies_total_progress_2}
        \end{equation}
        
    \end{enumerate}
\end{corollary}

We give the proof in the appendix (\Cref{proof:test_implies_distinguishing}).

\section{Warm-up: Proof of~\Cref{lem:warmup-overview}}
\label{sec:2^r}

In this section we prove the following: 

\begin{lemma}[\Cref{lem:warmup-overview} with a dependence on $N$]  
\label{lem:regular implies biased}
Let $\mathcal{A}$ be a $t$-parallel $d$-round quantum algorithm and $f : \zo^N \to [0,1]$ denote its acceptance probabilities. If $\mathcal{A}$ is $\eta$-regular where 
\[ \eta \le (dt\ln(N/\delta))^{-\Omega(2^d)},\] 
then $\min\{\Pr[\Acc_f],\Pr[\Rej_f]\}\le \delta$. 
    \end{lemma}

We will show in~\Cref{sec:removing dependence on $N$} how the dependence on $N$ in the regularity parameter of~\Cref{lem:regular implies biased} can be generically removed by setting $N = \poly(2^{td},1/\delta)$, thereby yielding~\Cref{lem:warmup-overview}. %

\subsection{Proof overview for~\Cref{lem:regular implies biased}}

\violet{Let $W_1(x),\ldots,W_N(x)$ be $\mathcal{A}$'s query weights, which we break down into their individual round contributions $W_i(x) = w_i^{(1)}(x) + \cdots + w_i^{(d)}(x).$}
Define %
 \[
        \gamma^* 
        \coloneqq
        O\!\left(\frac{1}{d^4t\ln(1/\delta)}\right).
    \]
and partition $\zo^N$ into $\mathsf{Good}\sqcup \mathsf{Bad}$ where 
\begin{align*}
\mathsf{Bad}
&:=
\{x \in\zo^N \colon \exists\, r\in[d],\exists\, i\in[N]\ \text{such that}\ w^{(r)}_i(x) \ge \gamma^*\} \\
\mathsf{Good}
 &:= \zo^N \setminus \mathsf{Bad}. 
 \end{align*}

We will prove the following lemmas:

\begin{lemma}[Good inputs are biased]
\label{lem:good inputs are biased}
$\Pr[\Acc]\cdot \Pr[\Rej\cap \mathsf{Good}] \le \delta$.  
\end{lemma} 

\begin{lemma}[Regularity $\Rightarrow$ Few bad inputs]
\label{lem:regularity implies few bad inputs} 
$\Prx[\mathsf{Bad}] \le \delta.$
\end{lemma} 

\Cref{lem:regular implies biased} follows as an immediate consequence since 
    \begin{align*}
\Pr[\Acc]\cdot \Pr[\Rej] &= \underbrace{\Pr[\Acc]\cdot \Pr[\Rej \cap \mathsf{Good}]}_{\le\,\delta \text{ by~\Cref{lem:good inputs are biased}}} + \underbrace{\Pr[\Acc]\cdot \Pr[\Rej \cap \mathsf{Bad}].}_{\le\,\Pr[\mathsf{Bad}]\,\le\,\delta \text{ by~\Cref{lem:regularity implies few bad inputs}}} 
    \end{align*}

We prove~\Cref{lem:good inputs are biased} in~\Cref{sec:good inputs are biased} and then~\Cref{lem:regularity implies few bad inputs} in~\Cref{sec:regularity implies few bad inputs}. 

\paragraph{Corollary of Talagrand's convex-distance inequality.} 
The proofs of both lemmas rely on a simple corollary of Talagrand's inequality (\Cref{thm:Talagrand}). This corollary allows us to bound the  probability that a random $\bx\sim \zo^N$ is far from any large set $A\sse \zo^N$, even if we measure distance using an $x$-specific distance metric $d_{W(x)}$. We are allowed to do so as long as this distance measure doesn't put too much weight on any one coordinate (i.e.~$\|W(x)\|_{\infty}\leq \gamma$). 

\begin{corollary} 
    \label{cor:talagrand-linf} %
    Let $A \subseteq \zo^N$ be any set, and for each $x\in\zo^N$ let
    $W(x)\in\R_{\ge 0}^N$ satisfy $\lone{W(x)}\le t$.
    Then for every $\gamma,\eps>0$, %
    \begin{equation*}
        \Pr[A]\cdot \Pr\big[\linf{W(\bx)} \leq \gamma \text{ and }\dist_{W(\bx)}(\bx,A) \geq \eps\big]
        \leq
        \exp\!\left(-\Omega\left(\frac{\eps^2}{\gamma t}\right)\right).%
    \end{equation*}
\end{corollary}

\begin{proof}

Let
\[
E:=\big\{x\in\zo^N:\linf{W(x)}\le \gamma\ \text{and}\ \dist_{W(x)}(x,A)\ge\eps\big\}.
\]
For any $x\in E$, we consider its normalized weight vector  
\[
\tilde{W}(x) :=\frac{W(x)}{\sqrt{\gamma t}}
\]
and note that since $\lone{W(x)}\le t$ and $\linf{W(x)}\le \gamma$, we have
\[
\| \tilde{W}(x) \|_2^2
\;=\;
\frac{\| W(x)\|_2^2}{\gamma t}
\;\le\;
\frac{\|W(x)\|_{\infty}\|W(x)\|_1}{\gamma t}
\;\le\;1.
\]
Therefore $\tilde{W}(x)$ is included in the supremum in the definition of Talagrand's convex
distance $d_{\mathrm{Tal}}(x,A)$, and hence
\[
d_{\mathrm{Tal}}(x,A)
\;\ge\;
\min_{y\in A}\sum_{i \in \Delta(x,y)} \tilde{W}_i(x) 
\;=\;
\frac{1}{\sqrt{\gamma t}}\min_{y\in A}\sum_{i\in\Delta(x,y)} W_i(x) 
\;=\;
\frac{\dist_{W(x)}(x,A)}{\sqrt{\gamma t}}
\;\ge\;
\frac{\eps}{\sqrt{\gamma t}}.
\]
Since every $x\in E$ satisfies 
\[
d_{\mathrm{Tal}}(x,A)\ge \frac{\eps}{\sqrt{\gamma t}}, 
\]
we can apply~\Cref{thm:Talagrand} with
$t=\eps/\sqrt{\gamma t}$ to get that 
\[
\Pr[A]\cdot \Pr[E]
\le
\exp\!\left(-\frac{t^2}{4}\right)
=
\exp\!\left(-\frac{\eps^2}{4\gamma t}\right).
\]
This completes the proof.
\end{proof}

\subsection{Good inputs are biased: Proof of~\Cref{lem:good inputs are biased}}\label{sec:good inputs are biased} %

The proof of \Cref{lem:good inputs are biased} is a simple application of \Cref{cor:talagrand-linf}. \Cref{item:distinguishibility_constraint_cor} of \Cref{cor:quantum-model} tells us that every $x\in \Rej$ is far from $\Acc$ with distance measured according to $W(x)$, and every $x\in\mathsf{Good}$, by definition, has bounded $\| W(x)\|_\infty$. \Cref{cor:talagrand-linf} thus translates into a bound on $\Pr[\Acc]\cdot \Pr[\Rej \cap \mathsf{Good}]$. 

In more detail, for each $x \in \zo^N$ we define its round-averaged query weight vector $\overline W(x) \in \R^N_{\ge 0}$ where: 
\[
\overline W_i(x) := \Ex_{\br\sim [d]}\big[ w^{(\br)}_{i}(x)\big] = \frac{W_i(x)}{d}  \quad \text{for each $i\in [N]$}
\]
and note that 
\[   \dist_{\overline W(x)}(x,y) =  \Ex_{\br\sim [d]}\big[ \dist_{w^{(\br)}(x)}(x,y)\big]. \]

Since $t$-parallel query weights sum to at most $t$ (\Cref{eq:t_parallel_index_query_weights_sum_to_t}),%
\begin{equation}
    \lone{\overline W(x)}
=
\Ex_{\br\sim [d]}\Bigg[\sum_{i=1}^N w^{(\br)}_{i}(x)\Bigg] 
\leq
t.
\label{eq:l1}
\end{equation}
If $x\in\mathsf{Good}$, then $w^{(r)}_{i}(x)\le \gamma^*$ for every $r\in [d]$ and $i \in [N]$ and so
\begin{equation}
\linf{\overline W(x)}\le \gamma^*.
\label{eq:linfty}
\end{equation}
Furthermore, for any $x\in \Rej$ and $y\in \Acc$, \Cref{item:distinguishibility_constraint_cor} of \Cref{cor:quantum-model} tells us that 
\begin{align*}
   \sum_{r \in [d]}\sqrt{\dist_{w^{(r)}(x)}(x,y)} &\geq \Omega(1)\\
    \sqrt{\sum_{r \in [d]}\dist_{w^{(r)}(x)}(x,y)} \sqrt{\sum_{r \in [d]} 1} &\geq \Omega(1) &\tag{Cauchy-Schwarz}\\
    \frac{1}{d}\sum_{r \in [d]}\dist_{w^{(r)}(x)}(x,y) &\geq \Omega\left(\frac{1}{d^2}\right)
\end{align*}
Hence, $\dist_{\overline W(x)}(x,y) = \Ex_{\br\sim [d]}\big[  \dist_{w^{(\br)}(x)}(x,y) \big]
\ge \Omega\left(\frac{1}{d^2}\right)$. 
Equivalently, for any $x\in \Rej$, 
\begin{equation}  \dist_{\overline W(x)}(x,\Acc) \ge \Omega\left(\frac1{d^2}\right).
\label{eq:dist to A}
\end{equation}

With~\Cref{eq:l1,eq:linfty,eq:dist to A} in hand, we can now apply~\Cref{cor:talagrand-linf} with $A = \Acc$, $\gamma=\gamma^*$, and $\eps=\Omega\left(\frac{1}{d^2}\right)$ to conclude that: 
\[ 
    \Pr[\Acc]\cdot \Pr[\mathsf{Good}\cap \Rej] \leq \exp\left(-\Omega\left(\frac1{\gamma^* d^4 t}\right) \right) \le \delta,
\]
where the final inequality holds by our choice of $\gamma^*$.

\subsection{Regularity implies few bad inputs: Proof of~\Cref{lem:regularity implies few bad inputs}}

\label{sec:regularity implies few bad inputs}

The next claim converts the expectation bound on query weights implied by the assumption of regularity 
 into a high-probability bound. We do so inductively, in a round-by-round fashion. \Cref{lem:regularity implies few bad inputs} then follows easily via a union bound.  

For $r\in [d]$ and $\gamma > 0$, we define the sets 
\[ \mathsf{Good}^{(\le r)}(\gamma) \coloneqq \big\{ x \in \zo^N \colon w^{(r')}_i(x) \le \gamma \text{ for all $r' \in [r]$ and $i\in [N]$} \big\}, \]
and $\mathsf{Bad}^{(\le r)}(\gamma) \coloneqq \zo^N\setminus \mathsf{Good}^{(\le r)}(\gamma)$.

\begin{claim}[Inductive step: goodness is preserved w.h.p]
\label{claim:induction}
    Fix $r\in \{2,\ldots,d\}$ and $i\in [N]$. Suppose $\Ex_{\bx}[w^{(r)}_i(\bx)] \le \gamma_r$ and 
define%
\[ \gamma_{r-1}
         \coloneqq 
        O\!\left(\frac{\gamma_r^2}{r^4t\ln(N/\delta)}\right).\] 
Then 
\begin{align*}
     \Prx\big[  \mathsf{Bad}^{(\leq r)}(4\gamma_r)  \cap \mathsf{Good}^{(\le r-1)}(\gamma_{r-1})\big] \leq \delta.
\end{align*}

        \end{claim}

\begin{proof}
For each $x \in \zo^N$ and $r\in\{2,\ldots,d\}$, we define its round-averaged weight vector $\overline W^{(\le r-1)}(x)$ where: 
\[
\overline W^{(\le r-1)}_i(x):= \Ex_{\br'\sim [r-1]}\big[ w^{(\br')}_i(x)\big] \quad \text{for each $i\in [N]$}
\]
and note that 
\[   \dist_{\overline W^{(\le r-1)}(x)}(x,y) =  \Ex_{\br'\sim [r-1]}\big[ \dist_{w^{(\br')}_x}(x,y)\big]. \]

Since $t$-parallel query weights sum to $t$ (\Cref{eq:t_parallel_index_query_weights_sum_to_t}),
\begin{equation}
    \lone{\overline W^{(\le r-1)}(x)}
=
\Ex_{\br'\sim [r-1]}\Bigg[\sum_{i=1}^N w^{(\br')}_{i}(x)\Bigg] 
\leq
t.
\label{eq:l1norm}
\end{equation}
Since $x\in\mathsf{Good}^{(\le r-1)}(\gamma_{r-1})$, we have that $w_{i}^{(r')}(x)\le \gamma_{r-1} $ for every $r' \in [r-1]$ and $i \in [N]$, and so
\begin{equation}
\linf{\overline W^{(\le r-1)}(x)}\le \gamma_{r-1}. 
\label{eq:linftynorm}
\end{equation}

Fix $r$ and $i$ to those in the statement of the claim. Define $A\coloneqq  \{y : w_{i}^{(r)}(y) < 2\gamma_r\}$. Since we are assuming that $\Ex_{\by}[w^{(r)}_i(\by)] \le \gamma_r$, it follows from Markov's inequality that $\Pr [A] \geq \lfrac12.$

Furthermore, for any $y\in A$ and $x$ such that $w_{i}^{(r)}(x) \geq 4\gamma_r$, \Cref{item:progress_evolution_cor} of \Cref{cor:quantum-model} tells us that  
\begin{align*}
     \sum_{r'\in [r-1]}\sqrt{\dist_{w^{(r')}(x)}(x,y)}&\geq \Omega(\sqrt{\gamma_r})\\
      \sqrt{\sum_{r' \in [r-1]}\dist_{w^{(r')}(x)}(x,y)} \sqrt{\sum_{r' \in [r-1]} 1} &\geq \Omega(\sqrt{\gamma_r}) &\tag{Cauchy-Schwarz}\\
    \frac{1}{r}\sum_{r' \in [r-1]}\dist_{w^{(r')}(x)}(x,y) &\geq \Omega\left(\frac{\gamma_r}{r^2}\right)
\end{align*}
Hence, $\dist_{\overline W^{(\le r-1)}(x)}(x,y) = \Ex_{\br'\sim [r-1]}\big[  \dist_{w^{(\br')}(x)}(x,y) \big]
\ge \Omega\left(\frac{\gamma_r}{r^2}\right).$ 
Equivalently, for any $x$ such that $w_{i}^{(r)}(x) \geq 4\gamma_r$, 
\begin{equation}  \dist_{\overline W^{(\le r-1)}(x)}(x,A) \ge \Omega\left(\frac{\gamma_r}{r^2}\right).
\label{eq:dist to S}
\end{equation}

With~\Cref{eq:l1norm,eq:linftynorm,eq:dist to S} in hand, we can now apply~\Cref{cor:talagrand-linf} with $A=A$, $\gamma=\gamma_{r-1}$, and $\eps=\Omega\left(\gamma_r/r^2\right)$, along with the fact that $\Pr[A]\ge \frac1{2}$, to conclude: 
\begin{align*}
    \Prx\big[ w^{(r)}_i(\bx) \ge 4\gamma_r  \text{\ and\ } \bx \in \mathsf{Good}^{(\le r-1)}(\gamma_{r-1})\big] \leq \exp\left(-\Omega\left(\frac{\gamma_r^2}{\gamma_{r-1} r^4 t}\right) \right) \le \frac{\delta}{N},
\end{align*}
where the final inequality holds by our choice of $\gamma_{r-1}$. Finally, union bounding over all $i\in [N]$, we obtain 
\begin{align*}
    \Prx\big[  \mathsf{Bad}^{(\le r)}(4\gamma_r) \cap \mathsf{Good}^{(\le r-1)}(\gamma_{r-1})\big] \leq N\cdot \Prx\big[ w^{(r)}_i(\bx) \ge 4\gamma_r  \text{\ and\ } \bx \in \mathsf{Good}^{(\le r-1)}(\gamma_{r-1})\big] \leq  \delta,
\end{align*}
which completes the proof. 
\end{proof}

\subsubsection{Proof of~\Cref{lem:regularity implies few bad inputs} given~\Cref{claim:induction}}

Define thresholds
\begin{align*} 
\gamma_d &\coloneqq \gamma^*\\ 
\gamma_{r-1} &\coloneqq O\left(\frac{\gamma_r^2}{r^4t\ln(dN/\delta)}\right) \quad \text{for $r \in \{2,\ldots,d\}$}
\end{align*}
Note $\mathsf{Good} = \mathsf{Good}^{(\le d)}(\gamma_d)$. Note also that $\gamma_1 \le \cdots \le \gamma_d$ where 
\[ \gamma_1 \ge (dt \ln(N/\delta))^{-O(2^d)}, \]
and this quantity on the RHS matches the regularity assumption of~\Cref{lem:regular implies biased}.  By choosing constants appropriately, we can therefore ensure that $\eta \le \gamma_1 \le \cdots \le \gamma_d$, which in turn ensures that the assumption of~\Cref{claim:induction} is satisfied for all $r\in \{2,\ldots,d\}$.

Since first-round weights are independent of $x$ (\Cref{item:inital_condition_cor}), the expectation bound implied by regularity trivially translates into a worst-case bound: 
\[ \Ex_{\bx}[w^{(1)}_i(\bx)] \le \eta \text{\ for all $i\in [N]$} \ \Rightarrow\ \text{For all $x$, we have}\ w^{(1)}_i(x) \le \eta  \text{\ for all $i\in [N]$},\]
and so $\Pr[\mathsf{Good}^{(1)}(\gamma_1)] = 1$. Consequently, if $x \in \mathsf{Bad} = \mathsf{Bad}^{(\le d)}(\gamma_d)$, it must ``go from $\mathsf{Good}$ to $\mathsf{Bad}$ at some point", and so 
\[ \Pr[\mathsf{Bad}] \le \sum_{r=2}^d \Pr\big[\mathsf{Bad}^{(\le r)}(\gamma_r) \cap \mathsf{Good}^{(\le r-1)}(\gamma_{r-1})\big] \le \frac{\delta}{d-1}\cdot (d-1) = \delta, \]
where the final inequality is an application of~\Cref{claim:induction} with $\delta = \frac{\delta}{d-1}$.

\section{An improved bound: Proof of~\Cref{thm:rsquared-overview}}

In this section we prove:

\begin{theorem}[\Cref{thm:rsquared-overview} with a dependence on $N$]  
\label{thm:rsquared-bound}
Let $\mathcal{A}$ be a $t$-parallel $d$-round quantum algorithm and $f : \zo^N \to [0,1]$ denote its acceptance probabilities. If $\mathcal{A}$ is $\eta$-regular where 
\[ \eta \le 2^{-\Omega(d^2)}\cdot (t\log N)^{-\Omega(d)}\cdot \log(1/\delta)^{-1}, \] 
then $\min\{\Pr[\Acc_f],\Pr[\Rej_f]\}\le \delta$. 
    \end{theorem}

\violet{
Again, the statement of~\Cref{thm:rsquared-bound} does not quite match that of~\Cref{thm:rsquared-overview} due to the dependence on $N$, but as mentioned, in~\Cref{sec:removing dependence on $N$} we will show that this dependence can be generically removed by setting $N = \poly(2^{td},1/\delta)$.  In the setting of~\Cref{thm:rsquared-bound} this amounts to replacing $\log N$ with $\log(1/\delta)$, which then yields~\Cref{thm:rsquared-overview}.}

\pparagraph{Proof overview.}
As discussed in \Cref{subsec:overview-improved}, this proof will rely on the following notion of \emph{distance}: For any distribution $\mcD$ and $x,y$, we denote,
\begin{equation*}
    \dist_{\mcD}(x,y) \coloneqq \Prx_{\bS \sim \mcD}[x_{\bS} \neq y_{\bS}],
\end{equation*}
and similarly, for set $A$,
\begin{equation*}
    \dist_{\mcD}(x, A) \coloneqq \min_{y \in A} \dist_{\mcD}(x, y).
\end{equation*}

This definition aligns with \Cref{eq:dist-bound-layer,eq:final_condition_implies_total_progress_2}: For any $x\in \Acc$, there is some round $r\in [d]$ for which $$\dist_{\mcD_x^{(r)}}(x, \Rej) \geq \poly(1/r).$$
Similarly, \Cref{eq:dist-bound-layer} means that if $\mcD_x^{(r)}$ and $\mcD_y^{(r)}$ are ``far" (i.e. distinguishable by some test), then there must be some $r' \leq r$ for which $\dist_{\mcD_x^{(r')}}(x,y)$ is large.%

For the warm-up (\Cref{lem:warmup-overview}), we formalized the notion of $\mcD^{(r)}_x$ as being ``concentrated" by having a single coordinate $i \in [N]$ that commonly appears in $\bS \sim \mcD^{(r)}$. Our  proof of~\Cref{thm:rsquared-bound} requires a more general definition: %
\begin{definition}[$(m,\gamma)$-hit and spread] 
\label{def:hit-spread}
    For $m \in [N]$ and $\gamma \in [0,1]$, we say that a distribution $\mcD$ over subsets of $[N]$ is {\sl $(m,\gamma)$-hit} if 
    \begin{equation*}
        \Prx_{\bS \sim \mcD}[\bS \cap T \neq \emptyset] \geq \gamma\quad \text{for some $T\sse [N]$ of size $m$.}
    \end{equation*}
    We will also say a distribution is $(m, \gamma)$-\textsl{spread} if it is not $(m,\gamma)$-\textsl{hit}, meaning there is no $T$ of size $m$ for which the above inequality holds.
\end{definition}
\violet{The most technical portion of this proof connects these notions of distance and concentration, showing that for a distribution to separate large sets, it must be concentrated.}

\begin{restatable}[Large distance implies hitting]{lemma}{DistToHit}
    \label{lem:dist-to-hit}
    Let $A\sse \zo^N$ be any set, and for each $x\in \zo^N$ let $\mcD_x$ be a distribution over subsets of $[N]$ of size at most $t$. For any $m\in [N]$ and $0 < \gamma < \tau$, 
\begin{equation*}
        \Pr[A] \cdot \Pr\big[
            \dist_{\mcD_{\bx}}(\bx, A) \geq \tau\text{ and }
            \mathcal D_{\bx}\ \text{is }(m,\gamma)\text{-spread}\,
        \big] \leq \exp\paren*{-\Omega\paren*{\frac{(\tau-\gamma)^2m}{\gamma t}}}.
    \end{equation*}
\end{restatable}

Using this, we will prove the two main lemmas, the analogues of~\Cref{lem:regularity implies few bad inputs,lem:good inputs are biased} in our warm-up. First, we will define parameters,
\begin{equation*}
    m_1 \geq \cdots \geq m_d \quad\quad\text{and}\quad\quad \gamma_1 \leq \cdots \leq \gamma_d
\end{equation*}
and for each round $r\in [d]$, the sets
\begin{equation}
    \label{eq:def-bad-good}
    \Bad^{(r)} \coloneqq \set{x : \mcD_x^{(r)} \text{ is $(m_r,\gamma_r)$-hit}} \quad\text{and}\quad \Good^{(\leq r)} \coloneqq \zo^N \setminus \bigcup_{r' \in [r]}\Bad^{(r')}.
\end{equation}
Note the direction of monotonicity: At lower layers, we have larger values for $m$ and smaller values of $\gamma$. Both of these correspond to a more stringent constraint in order to be in $\Good$. Our first bound uses this monotonicity to bound the number of bad inputs:
\begin{restatable}[Regularity $\Rightarrow$ Few bad inputs]{lemma}{BoundBadStronger}
    \label{lem:most-inputs-good}
    For any $r\in \{ 2,\ldots,d\}$, if the regularity parameter $\eta$ satisfies
  $\eta \leq \gamma_r/(4m_r)$ where 
    \begin{equation*}
        \gamma_r \geq \Omega\paren[\Bigg]{\sum_{r' \in [r-1]} \gamma_{r'} \cdot (r-r')^2}, 
    \end{equation*}
    then
    \begin{equation*}
        \Pr\bracket*{\Bad^{(r)} \cap \Good^{(\leq r-1)}} \leq \binom{N}{m_r} \cdot \exp\paren*{-\Omega\paren*{\frac{\gamma_r m_{r-1} }{tr^3}}}
    \end{equation*}    
\end{restatable}

With~\Cref{lem:most-inputs-good} in hand, a bound on $\Pr[\Bad^{(\le d)}]$ follows easily via a union bound.
The second lemma then says that if this probability is small, the function must be biased: 
\begin{restatable}[\violet{Good inputs are biased}]{lemma}{GoodToBiased}
    \label{lem:bound-var}
    There is an absolute constant $c$ such that if 
    \begin{equation}
        \label{eq:sum-biased}
        \sum_{r \in [d]} (d-r+1)^2 \cdot \gamma_r \leq c,
    \end{equation}
    then
    \begin{equation*}
        \Pr\big[\Acc\big] \cdot \Pr\big[\Rej \cap \Good^{(\leq d)}\big] \leq \exp\paren*{-\Omega\paren*{\frac{m_d}{td^3}}}.
    \end{equation*}
\end{restatable}

\paragraph{Structure of this section.} We prove~\Cref{lem:dist-to-hit} in~\Cref{sec:dist-to-hit}, \Cref{lem:most-inputs-good} in~\Cref{sec:most-inputs-good}, and~\Cref{lem:bound-var} in~\Cref{sec:lem:bound-var}.  In~\Cref{sec:rsquared-bound} we put everything together to prove~\Cref{thm:rsquared-bound}.

\subsection{Large distance implies hitting: Proof of \Cref{lem:dist-to-hit}}
\label{sec:dist-to-hit}

To prove~\Cref{lem:dist-to-hit}, we analyze the natural greedy algorithm that builds a size-$m$ set $T$ that $\gamma$-hits $\mcD_x$, under the assumption that $\dist_{\mcD_x}(x,A) \geq \tau$. At each step, this greedy algorithm picks the variable that overlaps with the largest fraction of $\bS\sim \mcD_x$. The following helps quantify the progress we make at each step:

\begin{claim}%
    \label{claim:dist-to-heavy-element}
    Let $\mcD$ be a distribution supported on subsets of $[N]$ of size at most $t$. \violet{For any input $x \in \zo^N$ and set $A \subseteq \zo^N$, there exists an $i\in [N]$ such that }
    \begin{equation*}
        \Prx_{\bS \sim \mcD}[i \in \bS] \geq \frac{\violet{\dist_{\mcD}(x,A)^2}}{t \cdot \dtal(x,A)^2}.
    \end{equation*}
\end{claim}
The proof of \Cref{claim:dist-to-heavy-element} is analogous to that of \Cref{cor:talagrand-linf}. We give it in \Cref{subsec:dist-to-heavy-element}. For now, we show how to use it to prove \Cref{lem:dist-to-hit}.
\subsubsection{Proof of \Cref{lem:dist-to-hit} assuming \Cref{claim:dist-to-heavy-element}} 
We begin by initializing $T^{(0)} = \emptyset$. At each step $j = 1,2,\ldots$, we set
    \begin{equation*}
        i_j \coloneqq \argmax_{i \in [N]} \Prx_{\bS \sim \mcD_x} \bracket*{i \in \bS \text{ and }\bS \cap T^{(j-1)} = \emptyset},
    \end{equation*}
    and let $T^{(j)} \coloneqq T^{(j-1)} \cup \set{i_j}$. We stop at the first time step $m^{\star}$ at which $T^{(m^{\star})}$ $\gamma$-hits $\mcD_x$. We also define $\mcD^{(j)}$ to be the distribution over subsets $\bS'$ formed by first drawing $\bS \sim \mcD_x$ and setting
    \begin{equation*}
        \bS' \coloneqq \begin{cases}
        \bS &\text{if }\bS \cap T^{(j)} = \emptyset \\
        \emptyset&\text{otherwise.}
        \end{cases}
    \end{equation*}
    Given this notation, we can write
    \begin{align}
        \Prx_{\bS \sim \mcD_x}\bracket*{\bS \cap T^{(m^{\star})} \neq \emptyset} &= \sum_{j \in [m^{\star}]} \Prx_{\bS \sim \mcD_x} \bracket*{i_j \in \bS \text{ and }\bS \cap T^{(j-1)} = \emptyset} \notag \\
        &= \sum_{j \in [m^{\star}]} \Prx_{\bS \sim \mcD^{(j-1)}} \bracket*{i_j \in \bS} \notag \\
        &=\sum_{j \in [m^{\star}]} \argmax_{i \in [N]} \Prx_{\bS \sim \mcD^{(j-1)}} \bracket*{i \in \bS}. \label{eq:sum arg max}
    \end{align}
    We wish to apply \Cref{claim:dist-to-heavy-element} to conclude that each term in the above sum must be large. To do so, we observe that
    \begin{equation*}
        \dist_{\mcD^{(j)}}(x, A) \geq \dist_{\mcD_x}(x, A) - \dtv(\mcD_x, \mcD^{(j)})
    \end{equation*}
    Furthermore, we have defined $m^{\star}$ to be the first time step for which $T^{(m^{\star})}$ $\gamma$-hits $\mcD_x$. This means that $\dtv(\mcD_x, \mcD^{(j)}) \leq \gamma$ %
    for every $j < m^{\star}$. We can therefore apply~\Cref{claim:dist-to-heavy-element} to~\Cref{eq:sum arg max} to get: 
    \begin{equation*}
        \Prx_{\bS \sim \mcD_x}\bracket*{\bS \cap T^{(m^{\star})} \neq \emptyset} \geq m^{\star} \cdot \frac{\max\set{0,\dist_{\mcD_x}(x, A) - \gamma}^2}{t \cdot \dtal(x,A)^2}.
    \end{equation*}
    Finally, we observe that if $\dist_{\mcD_x}(x,A) \geq \tau$ and $\dtal(x,A) \leq  \sqrt{\frac{(\tau-\gamma)^2\,m}{\gamma t}}$, then for $m^{\star} = m$, the above will be at least $\gamma$, which means that $\mcD_x$ is $(m,\gamma)$-hit. Therefore,
    \begin{equation*}
        \Pr\bracket*{
            \dist_{\mcD_{\bx}}(\bx, A) \geq \tau\text{ and }
            \mathcal D_{\bx}\ \text{is }(m,\gamma)\text{-spread}
        } \leq \Pr\bracket*{\dtal(x,A) > \sqrt{\lfrac{(\tau-\gamma)^2\,m}{\gamma t}}},
    \end{equation*}
and Talagrand's convex-distance inequality (\Cref{thm:Talagrand}) gives the desired result.

\subsubsection{Proof of \Cref{claim:dist-to-heavy-element}}
\label{subsec:dist-to-heavy-element}

    Let $w \in \R^N_{\ge 0}$ be the vector $w_i \coloneqq \Prx_{\bS \sim \mcD}[i \in \bS]$. Since $\mcD$ is supported on subsets of size at most $t$, we have $\lone{w} \leq t$. Our goal is to lower bound $\linf{w}$. We first apply H\"older's inequality
    \begin{equation}
        \label{eq:Holder}
        \linf{w} \geq \frac{\ltwo{w}^2}{\lone{w}} \geq \frac{\ltwo{w}^2}{t},
    \end{equation}
    which reduces our task to that of lower bounding $\ltwo{w}$. To do so, we first show that \violet{$\dist_{w}(x,A) \geq \dist_{\mcD}(x,A)$. It suffices to show the inequality for every $y \in A$,
    \begin{equation*}
        \dist_{\mcD}(x,y) =  \Prx_{\bS\sim\mcD}[x_{\bS}\ne y_{\bS}] \le \sum_{i\in \Delta(x,y)} \Prx_{\bS\sim\mcD}[i\in \bS] = \sum_{i\in\Delta(x,y)}w_i = \dist_w(x,y).
    \end{equation*}}  
    Finally, let $\tilde{w} = \frac{w}{\ltwo{w}}$, and so \[ d_{\tilde{w}}(x,A) \geq \frac{\dist_{\mcD}(x,A)}{\ltwo{w}}.\] This normalization ensures  $\tilde{w}$ is included in the supremum in the definition of $\dt$ (recall \Cref{thm:Talagrand}), and so $d_{\tilde{w}}(x,A) \leq  \dt(x,A)$. Combining these, we get: 
    \begin{equation*}
       \frac{\dist_{\mcD}(x,A)}{\ltwo{w}} \leq d_{\tilde{w}}(x,A) \leq \dtal(x,A),
    \end{equation*}
This along with~\Cref{eq:Holder} completes the proof.

\subsection{Regularity implies few bad inputs: Proof of \Cref{lem:most-inputs-good}}
\label{sec:most-inputs-good}
Our proof will use the following method of combining $\mcD^{(1)}_x,\ldots,\mcD^{(r)}_x$ into a single distribution. This choice is made because the series $\sum_{i}^{r} 1/i^2$ converges for any choice of $r$ and also no individual term is smaller than $1/\poly(r)$. The convergence ensures that \Cref{claim:combine-dist-exp} holds, whereas the fact that no term is too small ensures the $r^3$ term in \Cref{lem:most-inputs-good} is mild.%

\begin{claim}
    \label{claim:combine-dist-exp}
     For any $x\in \zo^N$ and $r\in [d]$ we define $\mcD^{(\quads,r)}_x$ to be the product distribution
    \begin{equation}
        \label{eq:def-quad-dist}
        \mcD^{(\quads,r)}_x \coloneqq \paren*{\mcD_x^{(1)}}^{r^2} \times \paren*{\mcD_x^{(2)}}^{(r-1)^2} \times \cdots \times \paren*{\mcD_x^{(r-1)}}^{2^2}\times \paren*{\mcD_x^{(r)}}^{1^2}.
    \end{equation}
    Then for $y\in \zo^N$, 
    \begin{equation}
        \label{eq:compare-exp}
        \Prx_{\bS \sim \mcD^{(\quads,r)}_x}[x_{\bS} \neq y_{\bS}] \geq \lfrac{1}{3}\cdot \min\set*{1,\paren[\Bigg]{\sum_{r' \in [r]}\sqrt{\Prx_{\bS \sim \mcD^{(r')}_x}\bracket*{x_{\bS} \neq y_{\bS}}}}^2}.
    \end{equation}
\end{claim}
Note that the right-hand side of \Cref{eq:compare-exp} is the quantity we can bound using \Cref{eq:dist-bound-layer,eq:final_condition_implies_total_progress_2}. We will also need a simple proposition: 
\begin{proposition}
    \label{prop:combine-dist-hit} 
    Let $\gamma = \sum_{r' \in [r]} \gamma_{r'} \cdot (r-r' + 1)^2$. For all $x \in \zo^N$, 
    \begin{equation*}
        x \in \Good^{(\leq r)} \implies \mcD^{(\quads,r)}_x\text{ is }(m_r,\gamma)\text{-spread}.
    \end{equation*}
\end{proposition}

We prove \Cref{claim:combine-dist-exp} in \Cref{subsec:combine-dists} and \Cref{prop:combine-dist-hit} in \Cref{subsec:proof-of-combine-dist-hit}. For now, we use them to prove~\Cref{lem:most-inputs-good}. 

\subsubsection{Proof of~\Cref{lem:most-inputs-good} assuming~\Cref{claim:combine-dist-exp} and~\Cref{prop:combine-dist-hit}}

    In order for $x$ to be in $\Bad^{(r)}$, there must exist some $|T| = m_r$ for which $\Prx_{\bS \sim \mcD^{(r)}_x}[\bS \cap T \neq \emptyset] \geq \gamma_r$. Our proof will union bound over all $\binom{N}{m_r}$ choices of $T$.

    For any fixed $T$, define 
    \begin{equation*}
        A_T \coloneqq \set*{y: \Pr_{\bS \sim \mcD^{(r)}_y}\big[\bS \cap T \neq \emptyset\big] \leq \frac{\gamma_r}{2}} \quad\text{and}\quad  \Bad_T \coloneqq \set*{x: \Pr_{\bS \sim \mcD^{(r)}_x}\big[\bS \cap T \neq \emptyset\big] \geq \gamma_r}.
    \end{equation*}
    Our goal will be to show that each $\Bad_T$ is small enough that we can afford a union bound over all choices of $T$. We begin by showing that $A_T$ is large. Using the $\eta \le \gamma_r/(4m_r)$ regularity assumption,
    \begin{equation*}
        \Ex_{\by \sim \zo^N}\bracket*{\Ex_{\bS \sim \mcD^{(r)}_{\by}}\bracket*{\abs*{\bS \cap T}}} = \sum_{i \in T}\Ex_{\by \sim \zo^N}\bracket*{\Prx_{\bS \sim \mcD^{(r)}_{\by}}\bracket*{i \in \bS}} \leq \eta \cdot \abs*{T} \leq \frac{\gamma_r}{4}.
    \end{equation*}
    Therefore, by Markov's inequality, $\Pr[A_T] \geq \frac1{2}$.

    Next, we can apply \Cref{eq:dist-bound-layer} with the test $S \mapsto \Ind[S \cap T \neq \emptyset]$. This gives, for all $x \in \Bad_T$ and $y \in A_T$,
    \begin{equation*}
         \paren*{\sum_{r' \in [r-1]}\sqrt{\Prx_{\bS \sim \mcD^{(r')}_x}\bracket*{x_{\bS} \neq y_{\bS}}}}^2 \geq \Omega(\gamma_r).
    \end{equation*}
    By \Cref{claim:combine-dist-exp}, we have that
    \begin{equation*}
        \dist_{\mcD^{(\quads,r-1)}_x}(x, A_T) \geq c \gamma_r,
    \end{equation*}
    for some absolute constant $c > 0$. Using an appropriately large choice of constant in our assumption that 
    \begin{equation*}
        \gamma_r \geq \Omega\paren[\Bigg]{\sum_{r' \in [r-1]} \gamma_{r'} \cdot (r-r')^2}, 
    \end{equation*}

we can apply~\Cref{prop:combine-dist-hit} to get that
    \begin{equation*}
        x \in \Good^{(\leq r-1)} \implies \mcD^{(\quads,r-1)}_x\text{ is  }(m_{r-1},\lfrac{c}{2}\gamma_r)\text{-spread.} 
    \end{equation*}
  We now apply \Cref{lem:dist-to-hit} with $\tau = c\gamma_r$ and $\gamma = \frac{c}{2}\gamma_r$, along with the fact that $\mcD_x^{(\quads,r-1)}$ is supported on sets of size $O(t r^3)$:
    \begin{equation*}
        \Pr\bracket*{A_T} \cdot \Pr\bracket*{
           \Bad_T \cap \Good^{(\leq r-1)}} \leq \exp\paren*{-\Omega\paren*{\frac{\gamma_r m_{r-1} }{t r^3}}}. 
    \end{equation*}
     Finally, we observe that
    \begin{equation*}
        \Bad^{(r)}  \cap \Good^{(\leq r-1)} = \bigcup_{|T| = m_r} \Bad_T \cap \Good^{(\leq r-1)}.
    \end{equation*}
    The desired result follows from our earlier bound of $ \Pr\bracket*{A_T} \geq \frac1{2}$ for all choices of $T$ and a union bound over the $\binom{N}{m_r}$ choices of $T$.

\subsubsection{Proof of \Cref{claim:combine-dist-exp}}
\label{subsec:combine-dists}
\Cref{claim:combine-dist-exp} will follow straightforwardly from the following two propositions.
\begin{proposition} 
    \label{prop:exponential-sum}
     For any distributions $\mcD^{(1)},\ldots,\mcD^{(r)}$, and $x,y \in \zo^N,$
     \begin{equation*}
         \sum_{r' \in [r]} (r-r'+1)^2 \cdot \Prx_{\bS \sim \mcD^{(r')}_x}\bracket*{x_{\bS} \neq y_{\bS}} \geq \frac{6}{\pi^2} \cdot \paren[\Bigg]{\sum_{r' \in [r]}\sqrt{\Prx_{\bS \sim \mcD^{(r')}_x}\bracket*{x_{\bS} \neq y_{\bS}}}}^2
     \end{equation*}
\end{proposition}
\begin{proof}
    This is a straightforward application of Cauchy-Schwarz and convergence of $\sum_{i = 1}^{\infty} 1/i^2 = {\pi}^2/6$.
    \begin{align*}
        \paren*{\sum_{r' \in [r]}\sqrt{\Prx_{\bS \sim \mcD^{(r')}_x}\bracket*{x_{\bS} \neq y_{\bS}}}}^2 &= \paren*{\sum_{r'}\sqrt{(r - r' + 1)^{-2}} \cdot \sqrt{(r - r' + 1)^2\cdot\Prx_{\bS \sim \mcD^{(r')}_x}\bracket*{x_{\bS} \neq y_{\bS}}}}^2\\
        &\leq \paren*{\sum_{r'} (r-r'+1)^{-2}}\cdot \paren*{\sum_{r'}(r - r' + 1)^2\cdot\Prx_{\bS \sim \mcD^{(r')}_x}\bracket*{x_{\bS} \neq y_{\bS}}}\\
        &\leq \frac{\pi^2}{6} \cdot \sum_{r'} (r-r'+1)^2 \cdot \Prx_{\bS \sim \mcD^{(r')}_x}\bracket*{x_{\bS} \neq y_{\bS}}. \qedhere
    \end{align*}
\end{proof}

\begin{proposition}
    \label{prop:ind-sum}
    For any independent events $\bE_1, \ldots, \bE_m$,
    \begin{equation*}
        \Pr[\bE_1 \text{ or } \cdots \text{ or }\bE_m] \geq (1 - e^{-1}) \cdot \min\set[\Bigg]{1, \sum_{i \in [m]}\Pr[\bE_i]}.
    \end{equation*}
\end{proposition}
\begin{proof}
    Let $p_i = \Pr[\bE_i]$. Then,
    \begin{align*}
        \Pr[\bE_1 \text{ or } \cdots \text{ or }\bE_m] &= 1 - \prod_{i \in [m]}(1-p_i) \\
        &\geq 1 - e^{-\sum_{i \in [m]} p_i}.
    \end{align*}
    The desired result follows by the inequality $1 - e^{-P} \geq (1 - e^{-1}) \cdot \min\set{1, P}$ applied with $P = \sum_{i \in [m]} p_i$.
\end{proof}

We now combine these propositions.
\begin{proof}[Proof of \Cref{claim:combine-dist-exp}]
    First, we apply \Cref{prop:ind-sum} which gives
    \begin{equation*}
         \Prx_{\bS \sim \mcD^{(\quads,r)}_x}[x_{\bS} \neq y_{\bS}] \geq (1 - e^{-1}) \cdot \min\set*{1,\sum_{r' \in [r]} (r-r'+1)^2 \cdot \Prx_{\bS \sim \mcD^{(r')}_x}\bracket*{x_{\bS} \neq y_{\bS}}}.
    \end{equation*}
    The desired result then follows from \Cref{prop:exponential-sum} and the inequality $(1-e^{-1}) \cdot \frac{6}{\pi^2} \geq \frac1{3}$.
\end{proof}

\subsubsection{Proof of \Cref{prop:combine-dist-hit}}
\label{subsec:proof-of-combine-dist-hit}
    \begin{proof}
       Let $x \in \Good^{(\leq r)}$. We will show that 
        \begin{equation*}
            \Prx_{\bS \sim \mcD^{(\quads,r)}_x}[\bS \cap T \neq \emptyset] < \gamma
        \end{equation*}
        for any set $T$ of size $m_{r}$. Fix such a $T$.  
        Since $x \in \Good^{(\leq r)}$, it is not in $\Bad^{(1)},\ldots, \Bad^{(r)}$. Since $m_{r'}\ge m_r$ for all $r' \in [r]$, this means that
        \begin{equation*}
            \Prx_{\bS \sim \mcD^{(r')}_x}[\bS \cap T \neq \emptyset] < \gamma_{r'} \quad\quad\text{for all $r' \in [r]$}.
        \end{equation*}
    Then, by a union bound and the definition of $\mcD^{(\quads,r)}_x$ in \Cref{eq:def-quad-dist},
    \begin{align*}
         \Prx_{\bS \sim \mcD^{(\quads,r)}_x}[\bS \cap T \neq \emptyset] &\leq \sum_{r' \in [r]}(r-r'+1)^2 \cdot \Prx_{\bS \sim \mcD^{(r')}_x}[\bS \cap T \neq \emptyset] \\
         &< \sum_{r' \in [r]}(r-r'+1)^2 \cdot \gamma_{r'} = \gamma. \qedhere
    \end{align*}
\end{proof}

\subsection{Good inputs are biased: Proof of \Cref{lem:bound-var}}
\label{sec:lem:bound-var}
We can reuse many of the tools used in the proof of \Cref{lem:most-inputs-good}. %
For every $x\in \Rej$ and $y\in \Acc$, a combination of \Cref{eq:final_condition_implies_total_progress_2} and \Cref{claim:combine-dist-exp} gives
    \begin{equation*}
        \Prx_{\bS \sim \mcD_x^{(\quads, d)}}[x_{\bS} \neq y_{\bS}] \geq v  = \Omega(1) .
    \end{equation*}
That is, 
\begin{equation*}
    \dist_{\mcD_x^{(\quads, d)}}(x, \Acc) \geq v.
\end{equation*}
By \Cref{lem:dist-to-hit} and the fact that $\mcD_x^{(\quads, d)}$ is supported on sets of size $O(td^3)$,
    \begin{equation*}
        \Pr\big[\Acc\big] \cdot \Prx_{\bx}\Big[\bx\in\Rej\text{ and } \mcD_{\bx}^{(\quads, d)}\text{ is $(m_d, \lfrac1{2} v)$-spread}\Big] \leq \exp\paren*{-\Omega\paren*{\frac{m_d}{td^3}}}.
    \end{equation*}
    Finally \Cref{prop:combine-dist-hit} and \Cref{eq:sum-biased} together give that if $x \in \Good^{(\leq d)}$ then $\mcD_{x}^{(\quads, d)}$ is $(m_d, v/2)$-spread, giving the desired inequality.

\subsection{Putting things together: Proof of \Cref{thm:rsquared-bound}}
\label{sec:rsquared-bound}
    \paragraph{Setting parameters.} First we set the parameters $\gamma_1,\ldots, \gamma_d$ and $m_1,\ldots, m_d$. For some absolute constant $c$, we set
\[         \gamma_r = c^{-(d-r)} \quad\quad\text{and}\quad\quad m_r \coloneqq \log(1/\delta) \cdot (c td^3 \log N)^{d-r + 1} \cdot \prod_{r' = r+1}^d \frac{1}{\gamma_{r'}}.
   \] 
    Using these parameters, we set
    \begin{align*}
        \eta \coloneqq \frac{\gamma_1}{4m_1} &\geq \lfrac1{4} \cdot c^{-d} \cdot \paren*{\log(1/\delta) \cdot (ctd^3 \log N)^d \cdot \paren*{\frac{1}{\gamma_1}}^{d}}^{-1}\\
        &\geq 2^{-O(d^2)} \cdot (t \log N)^{-d} \cdot \log(1/\delta)^{-1}. \tag{Absorbing $(cd^3)^{-d}$ into $2^{-O(d^2)}$}
    \end{align*}

    \pparagraph{Bounding the size of the $\Bad$ sets.} We claim that
    \begin{equation}
        \Pr\bracket*{\Bad^{(r)} \cap \Good^{(\leq r-1)}} \leq \frac{\delta}{2d} \quad \text{for each $r\in [d]$,} \label{eq:each slice}
    \end{equation}
    where we take the convention that $\Good^{(\leq 0)} = \zo^N$ and otherwise follow the definitions in \Cref{eq:def-bad-good}. For $r = 1$, the distribution $\mcD^{(1)}_x$ does not depend on $x$ (\Cref{item:inital_condition_cor} of \Cref{cor:quantum-model}). Therefore, the assumption of $\eta$-regularity gives, for all $T$ and $x$,
    \begin{equation*}
        \Prx_{\bS \sim \mcD^{(1)}_x}[\bS \cap T \neq \emptyset] \leq \eta \cdot |T|.
    \end{equation*}
    For our choice of $\eta$, we therefore have that there is no $T$ of size $m_1$ that $\gamma_1$-hits $\mcD^{(1)}_x$, and so $\Bad^{(1)}$ is empty.

    For any $r \in \{2,\ldots,d\}$, we will apply \Cref{lem:most-inputs-good}. To do so, we need to verify that
    \begin{equation*}
        \gamma_r \geq \Omega\paren[\Bigg]{\sum_{r' \in [r-1]} \gamma_{r'} \cdot (r-r')^2},
    \end{equation*}
    which does indeed hold for a large enough choice of $c$ due to the convergence of the series $\sum_{k = 1}^\infty k^2/c^k$. Therefore,
    \begin{equation*}
          \Pr\bracket*{\Bad^{(r)} \cap \Good^{(\leq r-1)}} \leq  \binom{N}{m_r} \cdot\exp\paren*{-\Omega\paren*{\frac{\gamma_r m_{r-1}}{tr^3}}}.
    \end{equation*}
   For the RHS of the above to be at most $\frac{\delta}{2d}$ it suffices to ensure
    \begin{equation*}
         \exp\paren*{-\Omega\paren*{\frac{\gamma_r m_{r-1}}{tr^3}}} \leq \frac{\delta}{d \binom{N}{m_r}}.
    \end{equation*}
    Using the bound $\binom{N}{m} \leq N^m$ and solving for $m_{r-1}$, it in turn suffices to ensure
    \begin{equation*}
        m_{r-1} \geq \Omega\left(\frac{tr^3}{\gamma_r} \cdot \paren*{m_r \log N + \log(d/\delta)}\right),
    \end{equation*}
    which does indeed hold for our choice of parameters. Having shown~\Cref{eq:each slice}, a union bound over all $r \in [d]$ gives that
    \begin{equation*}
        \Pr\bracket*{\Good^{(\leq d)}} \geq 1 - \lfrac{\delta}{2}.
    \end{equation*}

    \pparagraph{$\Good$ inputs are biased.} This step is a direct application of \Cref{lem:bound-var}. First, we verify that, for our choice of parameters, \Cref{eq:sum-biased} holds. This is once again true due to convergence of the series $\sum_{k = 1}^\infty k^2/c^k$. Therefore, \Cref{lem:bound-var} gives
    \begin{equation*}
        \Pr\big[\Acc\big] \cdot \Pr\big[\Rej \cap \Good^{(\leq d)}\big] \leq \exp\paren*{-\Omega\paren*{\frac{m_d}{td^3}}},
    \end{equation*}
    which is at most $\delta/2$ by our choice of $m_d$. 

We have shown that 
\[ \Pr\big[\Acc\big]\cdot \Pr\big[\Rej\big] \le \Pr\big[\Acc\big]\cdot \Pr\big[\Rej \cap \Good^{(\le d)}\big] + \Pr\big[\Bad^{(\le d)}\big] \le \lfrac{\delta}{2} + \lfrac{\delta}{2} = \delta,  \]
and this completes the proof of~\Cref{thm:rsquared-bound}.

\section{Incorporating $\eps$-dependence and removing $N$-dependence}
\label{sec:clean-up}

\violet{ 
In this section we show how statements such as~\Cref{lem:regular implies biased} and~\Cref{thm:rsquared-bound} can be generically improved in two ways: 

\begin{enumerate}
    \item We show how they imply a version that bounds $\eps$-sized deviations rather than the constant-sized ones. That is, we get a version where $\delta \coloneqq \Pr[\abs*{f(\bx) - \Ex[f(\bx)]} \geq \eps]$ at the cost of a $\poly(1/\eps)$-overhead in the regularity parameter. (Note that if $\min\{\Pr[\Acc_f],\Pr[\Rej_f]\} \ge \delta$ then either $\Acc_f$ or $\Rej_f$ witnesses the fact $\Pr\big[|f(\bx)-\E[f(\bx)]|\ge \frac1{6}\big] \ge \delta$.) 
    \item We then show how any dependence on $N$ in the regularity parameter can be removed by setting $N = \poly(2^{td},1/\eps,1/\delta)$. This essentially follows from the result of~\cite{DFKO07} discussed in the introduction, though we cannot quite apply it as a blackbox in our setting. 
\end{enumerate}
}

\subsection{Incorporating $\eps$-dependence}
\label{sec:eps dependence}
\begin{claim}[Bounding small deviations]
    \label{claim:eps-dependence}
    Suppose that, for some function $\eta(d,t,\delta)$, every $t$-parallel $d$-round quantum algorithm with acceptance probability $f:\zo^N \to[0,1]$ has a coordinate $i \in [N]$ for which
    \begin{equation*}
        \Ex_{\bx}[W_i(\bx)] \geq \eta(d,t,\delta) \quad\quad\text{where } \delta \coloneqq \min\big\{{\Pr[\Acc_f]}, \Pr[\Rej_f]\big\}.
    \end{equation*}
    Then, for every $\eps > 0$, every such algorithm also has a coordinate $i\in[N]$ for which
    \begin{equation*}
        \Ex_{\bx}[W_i(\bx)] \geq \eta(d,t,\delta,\eps)\quad\quad\text{where } \delta \coloneqq \Prx_{\bx}\bracket*{\abs*{f(\bx) - \Ex[f(\bx)]} > \eps},
    \end{equation*}
    and we define
    \begin{equation*}
        \eta(d,t,\delta,\eps) \coloneqq \Omega(\eps^2) \cdot \eta(d,O(t/\eps^2),\min\{\delta/2,\eps/2\}).
    \end{equation*}

\end{claim}
\Cref{claim:eps-dependence} follows from standard success amplification. The following inequality will be useful.
\begin{proposition}
    \label{prop:interval-to-mean}
    Let $\bp$ be a random variable supported on $[0,1]$ and $\eps, \delta$ satisfy
    \begin{equation*}
         \delta = \Pr[|\bp - \mu| \geq \eps] \quad\quad\text{where $\mu \coloneqq \Ex[\bp]$}.
    \end{equation*}
    Then, there are thresholds $b-a = \eps/2$ for which
    \begin{equation*}
        \min\{\Pr[\bp \leq a], \Pr[\bp \geq b]\} \geq \min\{\eps/2, \delta/2\}.
    \end{equation*}   
\end{proposition}
\begin{proof}
    Either $\Pr[\bp \leq \mu - \eps] \geq \delta/2$ or $\Pr[\bp \geq \mu + \eps] \geq \delta/2$. Without loss of generality, we assume it is the former. Given that, we set $a = \mu-\eps$ and $b = \mu-\eps/2$ which have the desired difference. Then, defining %
    $\delta_{\mathrm{large}} \coloneqq \Pr[\bp \geq b]$, we can upper bound the mean as
    \begin{equation*}
        \mu \leq (1 - \delta_{\mathrm{large}}) \cdot (\mu-\eps/2) + \delta_{\mathrm{large}} \cdot 1 \leq \mu - \eps/2 + \delta_{\mathrm{large}},
    \end{equation*}
    which requires $\delta_{\mathrm{large}} \geq \eps/2$. These thresholds therefore satisfy the desired inequality
\end{proof}

\begin{proof}[Proof of \Cref{claim:eps-dependence}]
    Let $f:\zo^N \to [0,1]$ be the acceptance probability of a $d$-round $t$-parallel quantum algorithm satisfying
    \begin{equation*}
         \delta \coloneqq \Prx_{\bx}\bracket*{\abs*{f(\bx) - \Ex[f(\bx)]} > \eps}.
    \end{equation*}
    Then, by \Cref{prop:interval-to-mean}, there are thresholds $b - a = \eps/2$ for which
    \begin{equation*}
        \min(\Pr[f(\bx) \leq a], \Pr[f(\bx) \geq b]) \geq \min\{\eps/2, \delta/2\}.
    \end{equation*}
    We will now construct an amplified version of the original quantum query algorithm and apply the hypothesis of \Cref{claim:eps-dependence} to this amplified variant. For $k = O(1/\eps^2)$ the amplified algorithm runs $k$ copies of the original algorithm in parallel, and for $\ell \coloneqq \floor{k(a + b)/2}$, outputs whether at least $\ell$ of those copies would output $1$. This is a $kt$-parallel $d$-round quantum query algorithm since its final state can be computed as
    \begin{equation*}
        \ket{\Psi(x)}^{\otimes k}\coloneqq
            (\mcO_x^{\otimes t})^{\otimes k}U_{d}^{\otimes k}(\mcO_x^{\otimes t})^{\otimes k}\cdots U_2^{\otimes k}(\mcO_x^{\otimes t})^{\otimes k}U_1^{\otimes k}\ket{1}.
    \end{equation*}
  
    The resulting algorithm has acceptance probability
    \begin{equation*}
        F(x) \coloneqq \Prx_{\bz \sim \mathrm{Bin}(k, f(x))}[\bz \geq \ell].
    \end{equation*}
    Chebyshev's inequality gives that if $f(x) \geq b$ then $F(x) \geq 2/3$ and if $f(x) \leq a$ then $F(x) \leq 1/3$. Therefore,
    \begin{equation*}
        \min\big\{{\Pr[\Acc_F]}, \Pr[\Rej_F]\big\} \geq \min\{\eps/2, \delta/2\},
    \end{equation*}
    which, by the hypothesis, implies that this amplified algorithm has a coordinate with expected query weight at least $\eta(d, kt, \min\{\eps/2, \delta/2\})$. The desired result follows from the observation that each expected query weight of the amplified algorithm is at most a factor of $k$ larger than the expected query weights in the original algorithm.
    \end{proof}

\violet{\paragraph{Applying~\Cref{claim:eps-dependence}.} We} can straightforwardly apply \Cref{claim:eps-dependence} to both our conjecture and~\Cref{thm:rsquared-bound}.
\begin{claim}[Small deviation version of our conjecture] \label{claim:eps-dependent-conjecture}
    If \Cref{conj:our conjecture} is true, for any $\eps > 0$, any $t$-query quantum algorithm with acceptance probability $f : \zo^N \to [0,1]$ has a variable $i \in [N]$ with query weight at least
    \begin{equation*}
         \Ex_{\bx}[W_i(\bx)] \ge \poly\paren*{\frac{\eps \delta}{t}} \quad\quad\text{ where }\delta \coloneqq \Pr[\abs*{f(\bx) -\Ex[f(\bx)]} \geq \eps].
    \end{equation*}
\end{claim}

We state the corollary of \Cref{thm:rsquared-bound} in the contrapositive so we can more easily apply it to prove the special case of the simulation conjecture, \Cref{thm:d squared intro}.
\begin{corollary}[Consequence of \Cref{thm:rsquared-bound} and \Cref{claim:eps-dependence}]
\label{cor:rsquared with eps dependence}
    Let $\mathcal{A}$ be a $t$-parallel $d$-round quantum algorithm and $f : \zo^N \to [0,1]$ denote its acceptance probability. If $\mathcal{A}$ is $\eta$-regular where 
\[ \eta \le 2^{-\Omega(d^2)}\cdot ((t\log N)/\eps)^{-\Omega(d)}\cdot \log(1/\delta)^{-1}, \] 
then $\Pr[\abs*{f(\bx) -\Ex[f(\bx)]} \geq \eps]\le \delta.$
\end{corollary}

\subsection{Removing $N$-dependence}
\label{sec:removing dependence on $N$}

We will need the following result of~\cite{DFKO07}: 
\begin{theorem}[\cite{DFKO07}]
    \label{thm:DFKO}
    Let $f : \zo^N \to [0,1]$ be a degree-$D$ polynomial. For every $\eps > 0$ there is a set $J\sse [N]$ of size $2^{O(D)}/\eps^2$ such that the function $f_{\sse J}:\zo^N\to [0,1]$,
    \[ f_{\sse J}(x) \coloneqq \Ex_{\br\sim\zo^N}\big[f(\br)\mid \br_J = x_J\big] \]
    satisfies $\| f - f_{\sse J}\|_2^2 \le \eps$. 
\end{theorem}

\begin{claim}[Removing dependence on $N$]
\label{claim:remove dependence on N}
    Suppose every $t$-parallel $d$-round quantum algorithm with acceptance probability $f : \zo^N\to [0,1]$ that is $\eta$-regular for some function $\eta(t,d,\eps,\delta,N)$ satisfies $\Pr[|f-\E[f]|\ge  \eps] \le \delta$. %
    Then every such algorithm that is $\eta$-regular where  \[ \eta \le \eta(t,d,\eps, \delta,2^{O(td)}/\eps^4\delta^2)\]   satisfies $\Pr[|f-\E[f]|\ge 2\eps]\le 2\delta$.  
\end{claim}

\begin{proof} 
Let $\mcA$ be a $t$-parallel $d$-round quantum algorithm that is $\eta$-regular where  
\begin{equation} \eta \le \eta(t,d,\eps,\delta,2^{O(td)}/\eps^4\delta^2),\label{eq:bound on regularity}
\end{equation}
Since $\mcA$ makes at most $td$ queries in total, its acceptance probability $f : \zo^N \to [0,1]$ has degree at most $2td$~\cite{BBCMdW01}. Applying~\Cref{thm:DFKO} with $f=f$, $D = 2td$, and $\eps = \eps^2\delta$, we get the existence of a set $J \sse [N]$,
\begin{equation} |J| \le \frac{2^{O(td)}}{\eps^4\delta^2} \quad \text{such that} \quad \| f - f_{\sse J}\|_2^2 \le \eps^2\delta. \label{eq:junta properties}
\end{equation}

Now consider the quantum query algorithm $\mcB$ defined as follows: On input $x\in \zo^N$,
\begin{enumerate}
    \item Sample $\br\sim \zo^N$ uniformly at random. 
\item Replace any queries to $x_i$ where $i\notin J$ with a query to $\br_i$. 
\end{enumerate} 
Note that $\mcB$ remains a $t$-parallel $d$-round algorithm (it is a mixture of restrictions of $\mcA$), and it is designed so that its acceptance probability $g : \zo^N \to [0,1]$ is given exactly by $f_{\sse J}$. Note also that $g$ depends only on the coordinates in $J$: there is a function $g' : \zo^J \to [0,1]$ such that $g(x) = g'(x_J)$ for all $x\in \zo^N$.

Furthermore, if $W_1(x),\ldots,W_N(x)$ are the query weights of $\mcA$, the query weights $\tilde{W}_1(x),\ldots,\tilde{W}_N(x)$ of $\mcB$ are given by: 
\begin{align*}\label{eq:quantum-junta}
\tilde{W}_i(x) = 
\begin{cases}
    \ds\Ex_{\br}\big[W_i(\br)\mid \br_J = x_J\big] & i \in J \\
    0 & otherwise. 
\end{cases}   
\end{align*}
Consequently $\Ex[\tilde{W}_i(\bx)] = \Ex[W_i(\bx)]$ if $i\in J$ and $\Ex[\tilde{W}_i(\bx)] = 0$ otherwise, and  so $\mcB$ inherits $\mcA$'s regularity (\Cref{eq:bound on regularity}).  By our bound on the size of $J$ (\Cref{eq:junta properties}), we further get that it is $\eta$-regular where 
\[ \eta \le \eta(t,d,\eps,\delta, |J|).\]
Since $g$ depends only on the variables in $J$, we can apply the assumption of our claim to $\mcB$ and $g$ with $N = |J|$ to get that $\Pr[|g-\E[g]|\ge \eps]\le \delta$. 

Finally, let $h = f-g$. By the triangle inequality and a union bound, 
\begin{align*}
\Pr\big[|f-\E[f]|\ge 2\eps\big] &\le \Pr\big[|g-\E[g]| \ge \eps\big] + \Pr\big[|h-\E[h]|\ge \eps\big] \\
&\le \delta + \frac{\| f-g \|_2^2}{\eps^2} \tag{Chebyshev} \\
&= \delta + \frac{\| f - f_{\sse J}\|_2^2}{\eps^2} \ \le\ 2\delta. \tag{$g \equiv f_{\sse J}$ and~\Cref{eq:junta properties}}
\end{align*}
This completes the proof. 
\end{proof}

\paragraph{Applying~\Cref{claim:remove dependence on N}.} We can now combine~\Cref{claim:remove dependence on N} and~\Cref{cor:rsquared with eps dependence} to finally prove~\Cref{thm:rsquared-overview}, and indeed, even the more general version with a dependence on $\eps$:

\begin{theorem}[\Cref{thm:rsquared-overview} with $\eps$-dependence]
\label{thm:rsquared with eps}
    Let $\mathcal{A}$ be a $t$-parallel $d$-round quantum algorithm and $f : \zo^N \to [0,1]$ denote its acceptance probability. If $\mathcal{A}$ is $\eta$-regular where 
\[ \eta \le 2^{-\Omega(d^2)}\cdot (t\log(1/\delta)/\eps)^{-\Omega(d)}, \] 
then $\Pr[\abs*{f(\bx) -\Ex[f(\bx)]} \geq \eps]\le \delta.$
\end{theorem}

For the same reason, combining~\Cref{claim:remove dependence on N} and~\Cref{lem:regular implies biased} yields~\Cref{lem:warmup-overview}.

\section{A regularity lemma and its implications}
\label{sec:regularity lemma}

In this section we prove our regularity lemma, repeated here for convenience:

\regularityoverview*

In~\Cref{sec:implications of regularity lemma} we then use it to show~\Cref{conj:our conjecture} implies the simulation conjecture, and likewise the strong version of~\Cref{conj:our conjecture} implies the strong version of the simulation conjecture. Similarly,~\Cref{thm:d squared intro} follows by combining~\Cref{thm:rsquared with eps} with~\Cref{lem:regularity lemma}.

\subsection{Proof of \Cref{lem:regularity lemma}}

For a quantum query algorithm $\mcA$ and a restriction $\pi$, we write $W_i(\pi)$ to denote the expected  query weight that $\mathcal{A}_{\pi}$ places on variable $i\in [N]$, where the expectation is over all  $\bx$'s that are consistent with $\pi$ (i.e.~all inputs to $\mcA_{\pi}$). Consider the classical query algorithm $\mathsf{Simulate}$ described in~\Cref{fig:regularity-alg}. It is a simple greedy algorithm that queries a variable of query weight larger than $\eta$ (if one exists), restricts $\mcA$ accordingly, and recurses.

\begin{figure}[h]
  \captionsetup{width=.9\linewidth}
    \algtext*{EndIf}
    \algtext*{EndFor}
  \begin{tcolorbox}[colback=white, arc=1mm, boxrule=0.25mm]
    \vspace{2pt}
    $\mathsf{Simulate}(\mcA,\eta,\delta, x)$:\vspace{6pt}
    
    \textbf{Input:} A $t$-query quantum algorithm $\mcA$ with acceptance probability $f : \zo^N\to [0,1]$ and  query weights $W_1(x),\ldots,W_N(x)$. Parameters $\eta$, $\delta$, and an input $x\in\zo^N$. \vspace{6pt}

    \textbf{Output:} An approximation of $f(x)$. \vspace{6pt}

    \begin{algorithmic}
    \State \textbf{Initialize:} $\pi_0 = \emptyset$ and $D=\poly(t,1/\eta,\log(1/\delta))$.
    \For{$j=1,\ldots, D$}
        \If{there exists an $i_j \in [N]$ such that $W_{i_j}(\pi_{j-1}) > \eta$}
          \State Query $x_{i_j}$ and let $b_j$ be the answer. 
          \State Let $\pi_{j} = \pi_{j-1} \cup \{x_{i_j}= b_j\}$ be the extension of $\pi_{j-1}$%
        \Else 
        \State Set $\pi_j = \pi_{j-1}$ and break. \vspace{6pt}
        \EndIf
      \EndFor
    \Return $\ds\Ex_{\bx}[f(\bx)\mid \bx\text{ is consistent with } \pi_j]$
    \end{algorithmic}
  \end{tcolorbox}
  \caption{A classical query algorithm for approximating a quantum query algorithm.}
  \label{fig:regularity-alg}
\end{figure}

We claim that $\mathsf{Simulate}$'s decision tree representation $T$ satisfies the conditions of \Cref{lem:regularity lemma}. Note that if we ever reach the ``Break" condition in \Cref{fig:regularity-alg}, then $W_i(\pi)\leq \eta$ for all $i\in[N]$, and so $\mcA_{\pi}$ is $\eta$-regular whenever $|\pi|< D$. It therefore suffices to show that $\Prx_{\bx}[|\bpi|\geq D] \leq \delta$, where $\bpi\sim T$ is the random path in $T$ induced by a uniform random $\bx\sim\zo^N$. 
  
To do so, let $\pi_j$ denote the truncation of a path $\pi$ to the first $j$ queries and consider the progress measure 
\[ P_j \coloneqq \sum_{i=1}^N W_i(\pi_j).\]  
 Note that $P_j$ is bounded between $[0,t]$ since $\mcA$ and its restrictions are $t$-query algorithms (\Cref{eq:t_parallel_index_query_weights_sum_to_t}). To analyze how it evolves as $j$ increases (i.e.~as the path grows down $T$), note that we can express $W_i(\pi)$ as 
\begin{align*}
     W_i(\pi)\coloneqq \begin{cases}
        0 & \text{if $i \in \pi$} \\
        \ds\Ex_{\bx}[W_i(\bx) \mid \bx \text{ is consistent with } \pi] & otherwise
        \end{cases}
\end{align*}
(see the proof of \Cref{claim:remove dependence on N}), and so 
\begin{equation}
\Ex_{\bb\sim \zo}\big[ W_i(\pi \cup \{ x_j = \bb\})\big] = 
\begin{cases}
0 & \text{if $i = j$}  \\
W_i(\pi) & \text{otherwise.} 
\end{cases}
\end{equation}
In words, when $x_j$ is queried, its query weight in both subfunctions drops to $0$, whereas for all other variables $i\ne j$, their query weights in the two subfunctions average to their original query weight. Consequently, with each query, our progress measure drops in expectation by exactly the weight of the queried variable:

\begin{proposition}[$\mathsf{Simulate}$ makes progress in expectation at each step]\label{prop:regularity-progress}  
    Let $\bP_j \coloneqq \sum_i W_i(\bpi_j)$ and let $\boldeta_j \coloneqq W_{\bi_j}(\mathbf{\bpi}_{j-1})$ be the query weight of the $j$th index queried. If we do not query an index, we say that $\boldeta_j=0$. Then, \[ \Ex\big[\bP_j\mid \bpi_{j-1}=\pi_{j-1}\big] =P_{j-1} - \eta_j.\] 
\end{proposition}

Since $P_j$ is bounded and drops in expectation as $j$ increases, this would suggest that $\bpi$ is unlikely to be too long. This motivates us in defining the following martingale: 
\begin{align*}
    \bZ_j = \bP_j+\sum_{j'=1}^{j} \boldeta_{j'}. 
\end{align*}
Note that if $|\pi|\geq D$, then $Z_D = P_D + \sum_{j=1}^{D} \eta_j \geq D\eta$  since $\mathsf{Simulate}$ only queries a variable if its query weight $\eta_j$ is greater than $\eta$. Therefore, $Z_D -Z_0\geq D\eta -P_0\geq D\eta -t$, and so: 
\[ 
     \Pr[|\bpi|\geq D]\leq  \Pr[\bZ_D -\bZ_0 \geq D\eta-t].
\]
The remainder of the proof will therefore focus on bounding the probability that the difference $\bZ_D-\bZ_0$ is large. We do so by showing that $\{\bZ_j\}$ is a martingale with respect to the random variables $\{\bpi_j\}$. First note that $\bZ_j$ is indeed a function of $\bpi_0,\ldots,\bpi_j$ since $\bP_j$ and $\boldeta_{j'}$ are functions of these variables for all $j'\in[j]$. Also, $\bZ_j$ has finite expected value since $\bP_j$ and $\boldeta_{j'}$ are bounded. Furthermore, 
\begin{align*}\Ex[\bZ_j\mid\pi_0,\ldots,\pi_{j-1}] &=  \Ex[\bP_j \mid \pi_0,\ldots,\pi_{j-1}] +  \Ex\bigg[\sum_{j'=1}^{j} \boldeta_{j'}\biggm| \pi_0,\ldots,\pi_{j-1}\bigg]\\
    &= P_{j-1}-\eta_{j} + \sum_{j'=1}^{j} \eta_{j'} &\tag{by \Cref{prop:regularity-progress}}\\
    &=Z_{j-1}. &\tag{by definition}
\end{align*}
Thus, $\{\bZ_j\}$ is indeed a martingale with respect to $\{\bpi_j\}$ as desired. Finally, note that 
\begin{align*}
    |Z_j-Z_{j-1}|= |P_j -P_{j-1}+\eta_j| \leq 2t
\end{align*}
since each term is in $[0,t]$, and so we can apply Azuma--Hoeffding's (\Cref{thm:azuma}) to conclude
\begin{align*}
    \Pr[\bZ_D -\bZ_0 \geq D\eta-t] \leq \exp\left(-\Omega\left(\frac{(D\eta -t)^2}{Dt^2}\right)\right).
\end{align*}
For an appropriate choice of $D=\poly(t,1/\eta,\log(1/\delta))$, we achieve the desired bound of $\Pr[\bZ_D -\bZ_0 \geq D\eta-t]\leq \delta$. 

\begin{remark}[This proof robustifies]
    \label{remark:robust-regularity} The proof of \Cref{lem:regularity lemma} assumes that we pick a coordinate with query weight $> \eta$ at each step. For its algorithmic version (\Cref{subsec:algorithmic-regularity}), we will need the fact that the proof of~\Cref{lem:regularity lemma} goes through essentially unmodified under the following weaker guarantee: Whenever a coordinate with query weight $> \eta$ exists, we pick one with query weight $> \eta/2$.
\end{remark}

\subsection{Implications of~\Cref{lem:regularity lemma}}
\label{sec:implications of regularity lemma}

\begin{claim}[\Cref{conj:our conjecture} $\Rightarrow$ Simulation conjecture]
\label{claim:our conjecture implies simulation conjecture}
If~\Cref{conj:our conjecture} holds then the simulation conjecture holds. 
\end{claim}

\begin{proof}
Assume~\Cref{conj:our conjecture} and let 
    $\eta^* = \poly(\eps\delta/t)$ be the lower bound on query weights given by~\Cref{claim:eps-dependent-conjecture}. Let $\mcA$ be a $t$-query quantum algorithm and consider running the algorithm $\mathsf{Simulate}$ in~\Cref{fig:regularity-alg} on $\mcA$, $\eta^*$, and $\delta$. It represents a classical decision tree $T$ of depth $\poly(t,1/\eta^*,\log(1/\delta)) = \poly(t,1/\eps,1/\delta)$ such that $\mcA_{\pi}$ is $\eta^*$-regular for $\ge 1-\delta$ fraction of  root-to-leaf paths $\pi$ in $T$.  For  these $\pi$'s, \Cref{claim:eps-dependent-conjecture} tells us that $\Pr[|f_{\pi}(\bx)-\Ex[f_{\pi}]|\ge \eps] \le \delta$ and $T$'s leaf value is $\E[f_{\pi}]$. We can therefore conclude that $\Pr[|T(\bx)-f(\bx)|\ge \eps]\le 2\delta$. 
\end{proof}

As mentioned in the introduction, one can also consider a strong version of~\Cref{conj:our conjecture} with a $\polylog(1/\delta)$ dependence, which we now state formally: 

\begin{conjecture}[Strong version of~\Cref{conj:our conjecture}]
\label{conj:our conjecture strong}
    Let $\mcA$ be a $t$-query quantum algorithm and $f : \zo^N \to [0,1]$ denote its acceptance probability. There is a variable $i \in [N]$ such that
    \begin{equation*}
         \Ex_{\bx\sim\zo^N}[W_i(\bx)] \ge \poly\paren*{\frac{1}{t \log(\frac{1}{\delta})}} \quad\quad\text{ where }\delta \coloneqq \min\big\{{\Pr[f(\bx) \ge \lfrac{2}{3}]}, \Pr[f(\bx) \le \lfrac{1}{3}]\big\}.
    \end{equation*}
\end{conjecture}

By~\Cref{lem:regularity lemma}, this strong version implies a strong version of the simulation conjecture where the classical query complexity has a $\polylog(1/\delta)$ dependence rather than $\poly(1/\delta)$. The proof of this implication is essentially identical to that of~\Cref{claim:our conjecture implies simulation conjecture}.  The only difference is that under~\Cref{conj:our conjecture strong}, the lower bound of~\Cref{claim:eps-dependent-conjecture} improves from $\poly(\eps\delta/t)$ to $\poly(\eps/(t\log(1/\delta)))$, and so the depth of the decision tree that $\mathsf{Simulate}$ represents improves from $\poly(t,1/\eps,1/\delta)$ to $\poly(t,1/\eps,\log(1/\delta))$. 

For the same reasons,~\Cref{thm:d squared intro} follows straightforwardly from~\Cref{thm:rsquared with eps} via~\Cref{lem:regularity lemma}.%

\section{Implications for random-oracle separations}
In this section, we show that if $\PromiseBPP = \PromiseBQP$ and our conjecture holds, then there is an efficient classical oracle algorithm that simulates any $\BQP$ oracle circuit on a random oracle. We also show that if $\PromiseQNC = \PromiseQuasi$, then there is an efficient classical oracle algorithm that simulates any $\QNC$ oracle circuit on a random oracle. %
These proofs require our algorithmic regularity lemma (\Cref{lem:algorithmic regularity lemma overview}) that we prove in \Cref{subsec:algorithmic-regularity}. We then give the implications of that regularity lemma in \Cref{subsec:oracle-implications}.

\pparagraph{Notation of this section.}  Note that in previous sections, we used $x$ synonymously with $\mcO$. That is {\sl not} the case in this section. In this section, we consider quantum algorithms which on input $x\in\zo^n$, are allowed to query an oracle $\mcO \in \zo^N$ where $N$ is exponentially larger than $n$. We will write $C^{\mcO}(x)$ to denote the output of a quantum oracle circuit $C$ on input $\ket{x}$ and oracle $\mcO$. 

Most of our focus will be on the performance of $C$ as a function of the oracle $\mcO$, and we define suitable notation. For a restriction $\pi$, we use $C_{\pi}$ to denote the oracle circuit such that $(C_\pi)^{\mcO} = C^{\mcO_\pi}$. This oracle circuit can be constructed in time $\poly(|C|, |\pi|)$ and behaves as follows: Whenever $C$ makes a query to its oracle, $C_{\pi}$ first checks if that query is restricted by $\pi$. If so, it uses $\pi$ for the response. Otherwise, it queries the oracle.

We will frequently reason about the query weights of an oracle circuit and so define notation to do so. In this section we work in the fully adaptive model, where $t$ is the number of adaptive queries made by $C$. For any $i \in [N]$ and $k \in [t]$ we use $w_i^{(k)}(C^{\mcO})$ to denote the query weight of $C$ to the $i^{\text{th}}$ index of $\mcO$ in its $k^{\text{th}}$ query and $W_i(C^{\mcO}) = \sum_{k \in [t]} w_i^{(k)}(C^{\mcO})$ to refer to the corresponding total query weight (recall \Cref{def:query_weight_t_parallel}  applied to $1$-parallel queries). As in \Cref{def:regular} $C$ is $\eta$-\textsl{regular} if, for all $i \in [N]$ its total query weights satisfy $\Ex_{\bmcO}[W_i(C^{\bmcO})] \leq \eta$.

\subsection{Our algorithmic regularity lemma}
\label{subsec:algorithmic-regularity}
In this section, we prove the following.
\begin{lemma}[Algorithmic regularity lemma, formal version of \Cref{lem:algorithmic regularity lemma overview}]\label{cor:algorithmic_regularity_lemma}
If $\PromiseBPP = \PromiseBQP$, there is a classical randomized algorithm that takes as input a quantum oracle circuit $C$ and access to an oracle $\mcO \in \zo^N$, runs in time $\poly(|C|, 1/\eta, \log(1/\delta))$, and outputs a restriction $\pi$ consistent with $\mcO$, satisfying, with probability at least $1-\delta$ over a uniform $\bmcO$ and the randomness of the algorithm, that $C_{\pi}$ is $\eta$-regular.
\end{lemma}

To prove \Cref{cor:algorithmic_regularity_lemma}, we give an efficient way of executing \textsf{Simulate} from \Cref{fig:regularity-alg} provided $\PromiseBPP=\PromiseBQP$. The key step is to design an algorithm for finding indices with high query weight (\Cref{claim:find_heavy}). Using a standard branch-and-prune approach (\Cref{prop:heavy_hitters}), it suffices to design an algorithm to estimate \textsl{sums} of query weights; i.e. for any $S \subseteq [N]$, an algorithm that estimates $\sum_{i \in S} \Ex_{\bmcO}[W_i(C^{\bmcO})]$ (\Cref{claim:approx_set_sum}). To do so, we begin by designing an algorithm for estimating the average acceptance probability of $C^{\bmcO}$ over a random oracle $\bmcO$ (see \Cref{claim:approx_expectation}) and then show it can be instantiated to estimate sums of query weights.

We begin with the following folklore consequence of $\PromiseBPP = \PromiseBQP$.
\begin{claim}[Classical estimation of quantum acceptance probabilities]\label{claim:classical_estimation}
    If $\PromiseBPP = \PromiseBQP$, then there is a classical randomized algorithm that, given a description of a quantum circuit $C$ (with no oracle access), an accuracy parameter $\epsilon > 0$, and a confidence parameter $\delta > 0$, outputs an estimate $\hat{p}$ satisfying
    \begin{equation*}
        \Prx\!\left[\,\lvert\, \hat{p} - \Prx[C~\text{accepts}]\,\rvert \geq \epsilon\right] \leq \delta,
    \end{equation*}
    in time $\poly(|C|, 1/\epsilon, \log(1/\delta))$.
\end{claim}
\begin{proof}
    The canonical $\PromiseBQP$-complete problem is: Given a quantum circuit $C$, accept if its acceptance probability is at least $2/3$ and reject if it is less than $1/3$. Therefore, if $\PromiseBPP = \PromiseBQP$, this task can be done by a randomized polytime algorithm with failure probability $\leq \delta$ in time polynomial in the input circuit size and $\log(1/\delta)$. By standard amplification techniques, these thresholds of $2/3$ and $1/3$ can be replaced with $\tau + \eps$ and $\tau$ at the cost of a $\poly(1/\eps)$ overhead. The desired result then follows by binary search.
\end{proof}

In order to estimate the acceptance probability of a quantum oracle circuit, we need an efficient way of choosing a random oracle. For this,
we use a standard construction of $t$-wise independent hash families.
\begin{fact}[$t$-wise independent hash families~\cite{Joffe}, Corollary 3.34 of \cite{Vad12}]\label{cor:twise_hash}
For every $m, t \in \mathbb{N}$, there is a family of $t$-wise independent functions $\mcH = \{h : \zo^{m} \rightarrow \zo \}$ such that choosing a random function from $\mcH$ takes $O(tm)$ random bits, and evaluating a function from $\mcH$ takes $\poly(t , m )$ time.
\end{fact}

For a quantum query circuit $C$, let $\mu$ denote the average acceptance probability of $C$ over the uniform distribution of oracles,
\begin{equation*}
    \mu(C) \coloneqq \Prx_{\bmcO}[C^{\bmcO}~\text{accepts}].
\end{equation*}
We may approximate $\mu$ by making use of the above strongly explicit $t$-wise independent functions.
\begin{claim}[$\PromiseBPP = \PromiseBQP$ implies efficient approximation of acceptance probabilities]\label{claim:approx_expectation}
    If $\PromiseBPP = \PromiseBQP$, there exists a classical algorithm $\ApproxExpectation$ which  estimates the acceptance probability $\mu \coloneqq \Pr_{\bmcO}[C^{\bmcO}~\text{accepts}]$ of any quantum query circuit $C$, to tolerance $\tau$ and failure probability $\delta$,
    \begin{equation*}
        \Prx\lbr{\labs{\ApproxExpectation(C, \tau, \delta) - \mu(C)} \geq \tau} \leq \delta,
    \end{equation*}
    and runs in time $\poly(|C|, \frac{1}{\tau}, \log(\frac{1}{\delta}))$.
\end{claim}
\begin{proof}
    Suppose the circuit $C$ makes $t$ queries to $\mcO$. Since the acceptance probability of a $t$-query quantum circuit is a multilinear polynomial of degree at most $2t$ in the bits of $\bmcO$, it is identical under any $2t$-wise independent distribution. Thus, for any distribution $\mcU_{2t}$ that is $2t$-wise independent,
    \begin{equation*}
        \Prx_{\bmcO}[C^{\bmcO}~\text{accepts}] = \Prx_{\bmcV \sim \mcU_{2t}}[C^{\bmcV}~\text{accepts}].
    \end{equation*}
    We may equate $\bmcV \sim \mcU_{2t}$ with a $2t$-wise independent hash $\mcH_{2t}$ family over boolean functions $h : \zo^n \rightarrow \zo$ on $n = \log N$ bits. Note that by \Cref{cor:twise_hash}, sampling  $h \sim \mcH_{2t}$ takes $O(t \cdot \log N)$ random bits, and evaluating a function from $\mcH$ takes $\poly(t , \log N)$ time. Let $\mcB_{q}$ be the quantum algorithm which takes as input a seed $s$ of size $O(t \log N)$, samples $h_s \sim \mcH_{2t}$ with the seed, and simulates $C^\mcO$ by replacing each query to $\mcO$ with a call to $h_s$. $\mcB_{q}$ is a quantum algorithm (without queries) that runs in time $\poly(|C|, t^2, \log N)$.
    For any fixed seed $s$, the circuit $\mcB_q(s)$ is a quantum circuit (without queries) of size $\poly(|C|, \log N)$. By \Cref{claim:classical_estimation}, there is a classical algorithm $\mcB_c$ which, given a seed $s$, accuracy $\epsilon$, and failure probability $\delta'$, outputs an estimate of $\Pr[\mcB_q(s)~\text{accepts}]$ to within $\pm\epsilon$ with probability $\geq 1 - \delta'$, in time $\poly(|C|, 1/\epsilon, \log(1/\delta'))$.
    The algorithm $\ApproxExpectation$ proceeds as follows:
    \begin{enumerate}[label=(\roman*)]
        \item Set $m = \ceil*{ \frac{2\ln(4/\delta)}{\tau^2}}$ and $\delta' = \delta/(2m)$.
        \item For $j \in [m]$: sample $\bh_j \sim \mcH_{2t}$ and run $\mcB_c(\bh_j, \tau/2, \delta')$ to obtain an estimate $\bz_j$ of $\Pr[\mcB_q(\bh_j)~\text{accepts}]$.
        \item Output $\hat{\bmu} = \frac{1}{m} \sum_{j=1}^m \bz_j$.
    \end{enumerate}

    By a union bound, all estimates satisfy $|\bz_j - \Pr[\mcB_q(\bh_j)~\text{accepts}]| \leq \tau/2$ simultaneously with probability $\geq 1 - \delta/2$. Condition on this event. The values $p_j \coloneqq \Pr[\mcB_q(\bh_j)~\text{accepts}] \in [0,1]$ are i.i.d.\ (over seed choice) with expectation $\mu$. By Hoeffding's inequality,
    \begin{equation*}
        \Prx\!\left[\left|\tfrac{1}{m}\textstyle\sum_j p_j - \mu\right| \geq \tfrac{\tau}{2}\right] \leq 2\exp(-m\tau^2/2) \leq \frac{\delta}{2}.
    \end{equation*}
    Since $|\hat{\bmu} - \frac{1}{m}\sum_j p_j| \leq \tau/2$ on our conditioned event, the triangle inequality gives $\Pr[|\hat{\bmu} - \mu| \geq \tau] \leq \delta$.

    The total runtime is $m$ invocations of $\mcB_c$ on circuits of size $\poly(|C|, \log N)$, giving $\poly(|C|, \frac{1}{\tau}, \log(\frac{1}{\delta}))$.
\end{proof}

We next show how to estimate sums of query weights.
\begin{claim}[Estimating sums of query weights]\label{claim:approx_set_sum}
    If $\PromiseBPP = \PromiseBQP$, there exists a classical algorithm $\ApproxQueryWeights$ which takes as input an oracle circuit $C$ and an efficiently representable subset of indices\footnote{Our proof will always take $S$ to be an interval $[a,b]$.} $S \subseteq [N]$ and parameters $\eps, \delta$, and satisfies
    \begin{equation*}
        \Prx\bracket*{\abs*{\ApproxQueryWeights(C,S, \eps, \delta) - \overline{W}_S(C)} \geq \eps} \leq \delta \quad\quad\text{where}\quad \overline{W}_S(C) = \sum_{i \in S}\Ex_{\bmcO}[W_i(C^{\bmcO})]
    \end{equation*}
    and runs in time $\poly(|C|, \frac{1}{\eps}, \log(\frac{1}{\delta}))$.
\end{claim}
\begin{proof}
    Let $t$ denote the query complexity of the circuit $C$, which is upper bounded by $|C|$. For any round $k \in [t]$ and set $S$, we design a quantum oracle circuit $\tilde{C}^{(k, S)}$ so that on oracle $\mcO$, the acceptance probability of $(\tilde{C}^{(k,S)})^{\mcO}$ is exactly equal to $\sum_{i \in S}w^{(k)}_i(C^{\mcO})$. This is simple to construct given $C$ and an efficient representation of $S$: $\tilde{C}^{(k,S)}$ is the oracle circuit with final state given by the $k^{th}$ pre-query state, and measurement operator given by the projector onto the set $S$. 
    Then, we observe that
    \begin{equation*}
        \overline{W}_S(C) = \sum_{k \in [t]} \Ex_{\bmcO}\bracket*{\sum_{i \in S}w^{(k)}_i(C^{\bmcO})} = \sum_{k \in [t]}\mu(\tilde{C}^{(k,S)}),
    \end{equation*}
    which suffices since the right-hand side of the above can be estimated using $t$ invocations of \Cref{claim:approx_expectation}.
\end{proof}

We will apply the following algorithm for finding heavy hitters given access to an estimation oracle, which is a standard technique in learning theory and property testing:

\begin{proposition}[Heavy hitters via binary search {\cite{GL89}}; see also {\cite{KM93, ODBook}}]\label{prop:heavy_hitters}
    Let weights $W_j \in [0,1]$ for $j \in [N]$ satisfy $\sum_{j \in[N]} W_j \le M$. Suppose we have access to an estimation oracle $\textsc{Est}$ such that for any interval subset $S =[a,b] \subseteq [N]$, precision $\epsilon > 0$, and failure probability $\gamma > 0$, $\textsc{Est}(S, \epsilon, \gamma)$ returns an estimate $\hat{W}(S)$ of $W(S) \coloneqq \sum_{j \in S} W_j$ satisfying
    \[
        \Pr\bigl[\lvert \hat{W}(S) - W(S) \rvert \geq \epsilon\bigr] \leq \gamma.
    \]
    in time $T(\eps, \gamma)$. Then, for any threshold $\tau > 0$ and failure probability $\delta > 0$, there exists an algorithm running in time $\poly(M, \log N, T(\tau/4, \Omega((\delta \tau) / (M \log N))))$ and satisfies:
    \begin{enumerate}[label=(\roman*)]
        \item \label{item:heavy_hitters_1} If there exists $i^\star \in [N]$ with $W_{i^\star} \geq \tau$, then the algorithm returns some $i$ with $W_i \geq \frac{\tau}{2}$ with probability at least $1 - \delta$.
        \item \label{item:heavy_hitters_2} If no $i^\star \in [N]$ satisfies $W_{i^\star} \geq \frac{\tau}{2}$, then the algorithm returns $\perp$ with probability at least $1 - \delta$.
    \end{enumerate}
\end{proposition}
For completeness, we give the proof in the appendix (\Cref{appendix:proof_heavy_hitters}). We apply it to the vector $\overline{W}_i(C) = \Ex_{\bmcO}[W_i(C^{\bmcO})]$ using \Cref{claim:approx_set_sum} as the estimation procedure. Since $\sum_{i \in [N]} \overline{W}_i(C)$ is upper bounded by $|C|$, we immediately obtain the following.
\begin{corollary}[Finding a heavy query weight]\label{claim:find_heavy}
    If $\PromiseBPP = \PromiseBQP$, there exists a classical (randomized) algorithm $\FindHeavy$ which takes as input a quantum query circuit  $C$, a tolerance $\tau$, and failure probability $\delta$, and satisfies
    \begin{enumerate}[label=(\roman*)]
        \item If there is an $i^{\star}$ with $\overline{W}_{i^{\star}}(C) \geq \tau$, then $\FindHeavy(C, \tau, \delta)$ returns some $i$ satisfying $\overline{W}_{i}(C)  \geq \tau/2$ with probability at least $1-\delta$.
        \item If there is no $i^{\star}$ with $\overline{W}_{i^{\star}}(C)  \geq \tau/2$, then $\FindHeavy(C, \tau, \delta)$ returns $\perp$ with probability at least $1-\delta$,
    \end{enumerate}
    and runs in time $\poly(|C|, \frac{1}{\tau}, \log(\frac{1}{\delta}))$.
\end{corollary}

Using \Cref{claim:find_heavy} we prove the main result of this subsection.
\begin{proof}[Proof of \Cref{cor:algorithmic_regularity_lemma}]
    We run the algorithm $\mathsf{Simulate}$ (\Cref{fig:regularity-alg}), replacing the statement of the existence of a coordinate of high query weight ($\exists~ W_i(\pi_{j-1}) > \eta$) with a call to $\FindHeavy(C_{\pi_{j-1}}, \eta, \delta')$ from \Cref{claim:find_heavy}, where $\delta' = \frac{\delta}{2D}$ and $D = \poly(t, \frac{1}{\eta}, \log(\frac{1}{\delta}))$ is the depth bound from \Cref{lem:regularity lemma}. Note that this requires constructing the circuit $C_{\pi_{j-1}}$ which can be done in time $\poly(|C|, |\pi_{j-1}|)$. 
    
    If $\FindHeavy$ returns $\perp$, the algorithm terminates; otherwise, it queries the returned index and continues.

    By a union bound over all $D$ steps, all calls to $\FindHeavy$ succeed simultaneously with probability at least $1 - D \cdot \delta' \geq 1 - \frac{\delta}{2}$. 

    It remains to show that the algorithm terminates within $D$ steps with high probability. By \Cref{remark:robust-regularity}, the proof of \Cref{lem:regularity lemma} extends to the setting where, at each step, we select a coordinate with expected query weight $\geq \frac{\eta}{2}$ rather than $> \eta$. 

    Combining both failure events, $C_\pi$ is $\eta$-regular with probability at least $1 - \delta$. Since $|\pi_j|< D$ for each $j < D$, each of the $D$ iterations calls $\FindHeavy$ in time $\poly(|C|, D, \frac{1}{\eta}, \log(\frac{1}{\delta}))$, so the total runtime is $\poly(|C|,  \frac{1}{\eta}, \log(\frac{1}{\delta}))$.
\end{proof}

The exact same proof as that of~\Cref{cor:algorithmic_regularity_lemma} gives the following.
\begin{claim}
\label{claim:QNC-regularity}
    If $\PromiseQNC \subseteq \PromiseQuasi$, there is a classical randomized algorithm that takes as input a quantum oracle circuit $C$ and access to an oracle $\mcO \in \zo^N$, runs in time $\poly(n^{\polylog(n)}, 1/\eta, \log(1/\delta))$, and outputs a restriction $\pi$ consistent with $\mcO$, satisfying, with probability at least $1-\delta$ over a uniform $\bmcO$ and the randomness of the algorithm, that $C_{\pi}$ is $\eta$-regular.
\end{claim}

\subsection{Putting things together: Proofs of~\Cref{thm:random oracle unrelativized equivalence,thm:QNC equivalence,thm:implication-weak-conjecture}}
\label{subsec:oracle-implications}
We first show that if~\Cref{conj:our conjecture strong} (the strong version of~\Cref{conj:our conjecture} where the dependence on $\delta$ is only logarithmic) holds, then $\PromiseBPP = \PromiseBQP$ iff they are equal relative to a random oracle. %

\begin{theorem}[\Cref{thm:random oracle unrelativized equivalence} restated] %
Assuming \Cref{conj:our conjecture strong}, $\PromiseBPP^{\bmcO} \neq \PromiseBQP^{\bmcO}$ for a random oracle $\bmcO$ with probability $1$ if and only if $\PromiseBPP \neq \PromiseBQP$.
\end{theorem}
\begin{proof}
    The easy direction is essentially the same argument as in~\cite{FR99} which showed that if $\BPP \neq \BQP$ then they are also unequal relative to a random oracle with probability $1$. We use the promise version of this statement. Briefly, if $\PromiseBPP \neq \PromiseBQP$, there is some promise problem $L \in \PromiseBQP$ but not $\PromiseBPP$. This problem is clearly in $\PromiseBQP^{\mcO}$ for any oracle $\mcO$. Furthermore, an argument of \cite{BG81} (originally for languages but easily extending to promise problems) shows that if $L \notin \PromiseBPP$, then it is also not in $\PromiseBPP^{\bmcO}$ with probability $1$ over the randomness of $\bmcO$.

    The harder direction is to show that if $\PromiseBPP = \PromiseBQP$ then $\PromiseBPP^{\bmcO} = \PromiseBQP^{\bmcO}$ with probability one over a random oracle $\bmcO$. We use \Cref{cor:algorithmic_regularity_lemma} and \Cref{conj:our conjecture strong} to show this holds. Let $Q$ be any efficient quantum oracle Turing machine. We will construct an efficient classical oracle Turing machine $A$ such that, with probability strictly greater than $0$ over $\bmcO$, $A^{\bmcO}$ solves any promise problem that $Q^{\bmcO}$ solves, meaning for all inputs $x$,
    \begin{equation}
    \label{eq:promise-equiv}
    \begin{split}
        \Pr_{Q}[Q^{\bmcO}(x) = 1] \geq 2/3 \implies \Pr_A[A^{\bmcO}(x) = 1] \geq 2/3,\\
        \Pr_Q[Q^{\bmcO}(x) =1] \leq 1/3 \implies \Pr_A[A^{\bmcO}(x) = 1]\leq 1/3.
    \end{split}
    \end{equation}
     Kolmogorov's zero–one law then implies $\PromiseBPP^{\bmcO} = \PromiseBQP^{\bmcO}$ with probability $1$ over the randomness of $\bmcO$.

     It suffices to show that, for any fixed $x$, the probability that \Cref{eq:promise-equiv} holds is at least $1 -  5^{-n}$. Then, a union bound over all possible strings $x$ gives a failure probability of at most $\sum_{n \in \N} 2^{n} \cdot 5^{-n} \leq 2/3$ which is strictly less than $1$. 

    For any $x$, there is an efficient classical algorithm that constructs a $\poly(n)$-sized oracle circuit $C$ such that the acceptance probability of $C^{\mcO}$ is within a constant ($1/100$ will suffice) of $Q^{\mcO}(x)$ for all oracles $\mcO$ \cite{Yao93,AW19}. Then, we use \Cref{cor:algorithmic_regularity_lemma} to efficiently find some restriction $\pi$ such that, with probability at least $1 - 1/3 \cdot 5^{-n}$ over a random oracle $\bmcO$, the circuit $C_{\pi}$ is $\eta = 1/\poly(n)$ regular. By \Cref{conj:our conjecture strong} and \Cref{claim:eps-dependence} this suffices to ensure that, for $f(\mcO) = \Pr[C^{\mcO}\text{ accepts}]$,
    \begin{equation*}
        \Prx_{\bmcO}[\abs*{f(\bmcO) - \mu(C_{\pi})} \geq 0.01\mid \text{$\bmcO$ consistent with $\pi$}] \leq \lfrac{5^{-n}}{3} %
    \end{equation*}
    Finally, we use \Cref{claim:approx_expectation} with accuracy $0.01$ and failure probability at most $1/3 \cdot 5^{-n}$ to estimate $\mu(C_{\pi})$ and have the classical algorithm output $\Ind[\mu(C_{\pi}) \geq 1/2]$. A union bound over the three failure probabilities gives that \Cref{eq:promise-equiv} holds with probability at least $1-5^{-n}$ over a random oracle $\bmcO$, as desired.
\end{proof}

The weaker form of our conjecture has a similar, but slightly weaker implication.
\begin{theorem}
    \label{thm:implication-weak-conjecture}
    Assuming \Cref{conj:our conjecture}, if $\PromiseBPP = \PromiseBQP$ then $\PromiseBQP^{\bmcO} \subseteq \PromiseHeurBPP^{\bmcO}$ with probability $1$ over a random oracle $\bmcO$.
\end{theorem}
Note that if $\PromiseBQP^{\bmcO} \subseteq \PromiseHeurBPP^{\bmcO}$ then $\BQP^{\bmcO} \subseteq \HeurBPP^{\bmcO}$. %
The proof of~\Cref{thm:implication-weak-conjecture} is identical to that of \Cref{thm:random oracle unrelativized equivalence} except we will only aim for a failure probability of $n^{-c}$ rather than $5^{-n}$ and can therefore use \Cref{conj:our conjecture} in place of \Cref{conj:our conjecture strong}.

Likewise, the proof of~\Cref{thm:QNC equivalence} is again identical to that of \Cref{thm:random oracle unrelativized equivalence} except we use \Cref{claim:QNC-regularity} in place of \Cref{cor:algorithmic_regularity_lemma} and \Cref{thm:rsquared with eps} in place of \Cref{conj:our conjecture strong}.

\section{Final comments}

After the submission of this manuscript, the authors obtained a few straightforward improvements of~\Cref{thm:d squared intro}. First,  the simulation can be made round preserving, where the classical algorithm is itself also a $d$-round parallel algorithm. Second, a tighter analysis results in an improved classical query complexity of $t^{O(d)}$. These proofs will appear in an accompanying  note~\cite{BDST26}.

\section{AI disclosure}
ChatGPT 5.4 was used for literature search, conversions of  handwritten notes into electronic ones, proof checking, and copyediting.

\section{Acknowledgments}

We thank the FOCS reviewers for helpful feedback and suggestions.

Guy, Carmen, and Li-Yang are supported by NSF awards 1942123, 2211237, 2224246, a Sloan Research
Fellowship, and a Google Research Scholar Award. Guy is also supported by a Jane Street Graduate
Research Fellowship and Omer Reingold’s Simons Foundation investigators award. Carmen is also supported by a Stanford
Computer Science Distinguished Fellowship, an NSF GRFP, Omer Reingold’s Simons Foundation investigators award, and a Jump Trading PhD Fellowship. 
This material is based upon work partially supported by the National Science Foundation under award No.~2016245 and 2440805 and by the Air Force Office of Scientific Research under grant agreement FA9550-21-1-0392. Jordan Docter is also supported in part by the Shoucheng Zhang Graduate Fellowship.
Any opinions, findings, and 
conclusions or recommendations expressed in this material are those of the authors and do not reflect 
the views of Jump Trading.

\bibliography{ref}
\bibliographystyle{alpha}

\appendix

\section{Aaronson--Ambainis implies~\Cref{conj:our conjecture}}
\label{sec:AA implies us}

\begin{claim}
    \label{claim:AA-to-us}
    If the Aaronson--Ambainis conjecture is true, then \Cref{conj:our conjecture} is also true.
\end{claim}

We will use the following simple relation between query weights and influences.
\begin{proposition}
    \label{prop:query-weight-to-inf}
Let $\mathcal{A}$ be a $d$-round quantum query algorithm with query weights $W_1(x),\ldots,W_N(x)$ and let $f: \zo^N \to [0,1]$ denote its acceptance probabilities. Then for each $i\in [N]$, 
\[ \Ex_{\bx \sim \zo^N}[W_i(\bx)] \ge \Omega\paren*{\frac{\Inf_i(f)}{d}} 
\] 
where $\Inf_i(f) \coloneqq \Ex_{\bx\sim\zo^N}\big[(f(\bx)-f(\bx^{\oplus i}))^2\big]$ and $\bx^{\oplus i}$ denotes $\bx$ with its $i$-th bit flipped. 
\end{proposition}

\begin{proof}
    For any $x \in \zo^N$, \Cref{eq:final_condition_implies_total_progress_2} gives that
    \begin{align*}
        \left|\sqrt{f(x)} - \sqrt{f(x^{\oplus i})} \right| &\leq 2 \sum_{r \in [d]} \sqrt{w_i^{(r)}(x)} \\
        &\leq 2 \sqrt{d}\sqrt{\sum_{r \in [d]} w_i^{(r)}(x)} \tag{Cauchy-Schwarz}\\
        &= 2 \sqrt{d}\sqrt{W_i(x)}.
    \end{align*}
    Next, observe that for any $p,q \in [0,1]$,
    \begin{equation*}
        (p-q)^2 =(\sqrt{p} + \sqrt{q})^2 (\sqrt{p} - \sqrt{q})^2  \leq 4(\sqrt{p} - \sqrt{q})^2.
    \end{equation*}
    Combining this with the prior equation completes the proof: 
    \begin{equation*}
        \Inf_i(f) \leq 4 \Ex\bracket*{\paren*{\sqrt{f(\bx)} - \sqrt{f(\bx^{\oplus i})}}^2} \leq 16d\cdot  \Ex\big[W_i(\bx)\big]. \qedhere
    \end{equation*}
\end{proof}

\begin{proof}[Proof of \Cref{claim:AA-to-us}]
    Let $\mcA$ be a $t$-query quantum algorithm with acceptance probabilities $f:\zo^N \to [0,1]$ and define
    \begin{equation*}
        \delta \coloneqq \min\big\{{\Pr[f(\bx) \ge \lfrac{2}{3}]}, \Pr[f(\bx) \le \lfrac{1}{3}]\big\}.
    \end{equation*}
    Since $f$ has degree at most $2t$~\cite{BBCMdW01}, the  Aaronson--Ambainis conjecture gives an $i \in [N]$ for which
    \begin{equation*}
        \Inf_i(f)
            \ge \poly\left(\frac{\Var(f)}{t}\right).
    \end{equation*}
    Next we relate $\Var(f)$ to $\delta$. 
    Writing $\Var(f) = \Ex[(f(\bx) - \mu)^2]$ where $\mu \coloneqq \Ex[f]$, we see that if $\mu \leq 1/2$ then $\Var(f) \geq \Pr[f(\bx) \geq 2/3] \cdot (1/6)^2$, and if $\mu \geq 1/2$ then $\Var(f) \geq \Pr[f(\bx) \leq 1/3] \cdot (1/6)^2$. In either case, $\Var(f) \geq \Omega(\delta)$. Since every $t$-query algorithm is trivially also a $t$-round algorithm: 
    \begin{align*}
\Ex_{\bx\sim\zo^N}[W_i(\bx)] &\ge \Omega\left(\frac{\Inf_i(f)}{t}\right) \tag{\Cref{prop:query-weight-to-inf} with $d = t$} \\
&\ge  \poly\left(\frac{\Var(f)}{t}\right)\tag{Aaronson--Ambainis} \\ 
&\ge \poly\left(\frac{\delta}{t}\right). \tag{$\Var(f)\ge \Omega(\delta)$}
    \end{align*}
This completes the proof. 
\end{proof}

\section{Hybrid method}
With respect to the $t$-parallel index and set query weights, a modified hybrid method bounds the distance of the final states of a quantum algorithm on differing oracle inputs. For a set $B \sse [N]$ we define $\Pi_{\cap B}$ as the projector given by the sum of the orthogonal projectors $\Pi_S^{(t)}$ for any $S\subseteq [N]$ which have nonempty intersection with $B$, such that
\begin{align*}
    \Pi_{\cap \Delta(x,y)}^{(t)} \coloneqq \sum_{S \cap \Delta(x,y) \ne \emptyset} \Pi_S^{(t)}.
\end{align*}
We write $\Pi_{\cap \Delta(x,y)}$ without the superscript $t$ when it is clear from context.

\begin{proposition}[Hybrid Method (parallel query variation) \cite{BBBV97}\label{prop:hybrid_method}]
    Consider a $t$-parallel $d$-round query algorithm with intermediate pre-query states $\psix{r}$ for each $r \in [d]$ and final state $\ket{\Psi(x)}$. For any two oracle inputs $x, y \in \zo^N$, the following properties hold:
    \begin{enumerate}[label=(\roman*)]
        \item Initial condition: $\pnorm{2}{\psix{1} - \psiy{1}} = 0$.
        \item Progress evolution, for any $r \in \set{2,\ldots, d}$: 
        \begin{align*}
            \pnorm{2}{\psix{r} - \psiy{r}} \leq \pnorm{2}{\psix{r-1} - \psiy{r-1}} + 2 \pnorm{2}{\Pi_{\cap \Delta(x,y)} \psix{r-1}},\\
            \shortintertext{\centering and}
            \pnorm{2}{\psix{r} - \psiy{r}} \leq \pnorm{2}{\psix{r-1} - \psiy{r-1}} + 2 \pnorm{2}{\Pi_{\cap \Delta(x,y)} \psiy{r-1}}.
        \end{align*}
        \item Distinguishability constraint: If a measurement $M_0$ distinguishes the final states such that $f(x) = \pnorm{2}{M_0 \ket{\Psi(x)}}^2$ and $f(y) = \pnorm{2}{M_0 \ket{\Psi(y)}}^2$, then 
        \begin{equation*}
            \pnorm{2}{\ket{\Psi(x)} - \ket{\Psi(y)}}^2 \geq (\sqrt{f(x)} - \sqrt{f(y)})^2 + (\sqrt{1-f(x)} - \sqrt{1-f(y)})^2.
        \end{equation*}
    \end{enumerate}
\end{proposition}

The argument follows the standard hybrid method closely \cite{BBBV97}, but we give it here for completeness.
\begin{proof}
    The initial condition follows by definition, as the starting state is independent of the oracle input. The proof of progress evolution follows from the triangle inequality and unitary invariance of the $\ell_2$-norm:
    \begin{align*}
        \pnorm{2}{\psirx - \psiry} &= \pnorm{2}{U_r (\mcO_x^{\otimes t} - \mcO_y^{\otimes t})\psix{r-1} + U_r\mcO_y^{\otimes t}\ltup{\psix{r-1}-\psiy{r-1}}} \\
        &\le \pnorm{2}{(\mcO_x^{\otimes t} - \mcO_y^{\otimes t})\psix{r-1}} + \pnorm{2}{\psix{r-1}-\psiy{r-1}} \\
        &\le 2\pnorm{2}{\Pi_{\cap \Delta(x,y)}\psix{r-1}} + \pnorm{2}{\psix{r-1}-\psiy{r-1}}.
    \end{align*}
    For the distinguishability constraint, we relate the Euclidean distance to the quantum fidelity, and subsequently to the classical Hellinger distance. By the monotonicity of fidelity under completely positive trace-preserving maps (such as the measurement $M_0$), we have:
    \begin{align*}
        \abs{\braket{\Psi(y)}{\Psi(x)}} \le \sqrt{f(x) f(y)} + \sqrt{(1-f(x))(1-f(y))}.
    \end{align*}
    Thus, the squared Euclidean distance is bounded by:
    \begin{align*}
        \pnorm{2}{\ket{\Psi(x)} - \ket{\Psi(y)}}^2 &= 2 - 2 \Re \braket{\Psi(y)}{\Psi(x)} \\
        &\ge 2 - 2\abs{\braket{\Psi(y)}{\Psi(x)}} \\
        &\ge 2 - 2\ltup{\sqrt{f(x) f(y)} + \sqrt{(1-f(x))(1-f(y))}} \\
        &= \ltup{\sqrt{f(x)} - \sqrt{f(y)}}^2 + \ltup{\sqrt{1-f(x)} - \sqrt{1-f(y)}}^2. \qedhere
    \end{align*}
\end{proof}

\section{Proof of \Cref{cor:quantum-model}}
\begin{proof}\label{proof:test_implies_distinguishing}
    Note that the weighted distance measures have a physical interpretation with respect to pre-query states at each round $r \in [d]$,
    \begin{align}
    \dist_{\mcD_x^{(r)}}(x,y) = \sum_{S \cap \Delta(x, y)\neq \emptyset} w_S^{(r)}(x) =  \pnorm{2}{\Pi^{(t)}_{\cap \Delta(x,y)} \psix{r}}^2, \label{eq:distance_vs_states_1}\\
     \shortintertext{\centering and}
    \dist_{\mcD_y^{(r)}}(x,y) = \sum_{S \cap \Delta(x, y)\neq \emptyset} w_S^{(r)}(y) = \pnorm{2}{\Pi^{(t)}_{\cap \Delta(x,y)} \psiy{r}}^2, \label{eq:distance_vs_states_2}
    \end{align}
    \Cref{item:inital_condition_cor} follows directly from the initial condition of \Cref{prop:hybrid_method}. We address \Cref{item:distinguishibility_constraint_cor} first.
     Applying the distinguishability constraint of \Cref{prop:hybrid_method},  
    \begin{align*}
        \pnorm{2}{\ket{\Psi(x)} - \ket{\Psi(y)}}^2 \ge \ltup{\sqrt{f(x)} - \sqrt{f(y)}}^2 + \ltup{\sqrt{1-f(x)} - \sqrt{1-f(y)}}^2  \ge \ltup{\sqrt{f(x)} - \sqrt{f(y)}}^2.
    \end{align*}
    By the initial condition and progress measure of \Cref{prop:hybrid_method}, 
    \begin{align*}
        \pnorm{2}{\ket{\Psi(x)} - \ket{\Psi(y)}} \le 2 \sum_{r \in [d]}\pnorm{2}{ \Pi_{\cap \Delta(x,y)} \psirx}.
    \end{align*}
    Combining these two inequalities and with \Cref{eq:distance_vs_states_1,eq:distance_vs_states_2}, we obtain
    \begin{equation}
        \ltup{2 \sum_{r \in [d]}\pnorm{2}{ \Pi^{(t)}_{\cap \Delta(x,y)} \psix{r}}}^2 \ge \ltup{\sqrt{f(x)} - \sqrt{f(y)}}^2. \label{eq:final_condition_implies_total_progress_1}
    \end{equation}
    The last inequality of \Cref{item:distinguishibility_constraint_cor} follows from \Cref{def:weighted_dist_subset} and the first follows from \Cref{eq:query_weights_vs_differentiating_set}.
    \Cref{item:progress_evolution_cor} follows by applying  \Cref{item:distinguishibility_constraint_cor} to the $r-1$-round quantum algorithm given by the first $r$ intermediate states whose final measurement is $\sum_{S \in \supp(\phi)} \Pi_S$ and 
    acceptance probability for each $x \in \zo^N$ given by,
    \begin{equation*}
        \pnorm{2}{\sum_{S \in \supp(\phi)} \Pi_S \psirx}^2 = \sum_{S \in \supp(\phi)} w_S^{(r)}(x) = \Ex_{\bS \sim \mcD^{(r)}_x}[\phi(\bS)].
    \end{equation*}
    \end{proof}

\section{Proof of \Cref{prop:heavy_hitters}}\label{appendix:proof_heavy_hitters}

The proof is a straightforward application of the branch-and-prune technique for finding heavy hitters~\cite{GL89, KM93, ODBook}, which we include for completeness.

\begin{proof}
    Set $L = \lceil \log N \rceil$, $\epsilon = \frac{\tau}{4}$, and $\gamma = \violet{O}(\frac{\delta \tau}{M L})$. The algorithm proceeds as follows.
    \begin{enumerate}
        \item Set the active list $\mcL_0 = \{[N]\}$.
        \item For each level $\ell = 0, 1, \ldots, L - 1$:
        \begin{enumerate}
            \item For each $S \in \mcL_\ell$, partition $S$ into two halves $S_1, S_2$ of size $\lfloor \frac{|S|}{2} \rfloor$ and $\lceil \frac{|S|}{2} \rceil$.
            \item For each half $S_b,~b \in \{1,2\}$, compute $\hat{W}(S_b) \gets \textsc{Est}(S_b, \epsilon, \gamma)$.
            \item Set $\mcL_{\ell+1} = \{ S_b : \hat{W}(S_b) > \frac{\tau}{2} \}$.
        \end{enumerate}
        \item In the last layer, only singletons will remain. For each such singleton $\{i\} \in \mcL_L$, compute $\hat{W}(\{i\}) \gets \textsc{Est}(\{i\}, \frac{\tau}{4}, \gamma)$. If $\hat{W}(\{i\}) \geq \frac{3\tau}{4}$, return $i$.
        \item If no singleton has high enough weight, return $\perp$.
    \end{enumerate}

    \emph{Correctness.} We condition on the event $\mcE$ that all calls to $\textsc{Est}$ return estimates within their prescribed precision. We first bound the branching. Any set $S$ that enters $\mcL_{\ell+1}$ satisfies $\hat{W}(S) > \frac{\tau}{2}$, so $W(S) > \frac{\tau}{2} - \epsilon > 0$. At each level, the active sets are disjoint subsets of $[N]$, so $|\mcL_\ell| \leq \frac{M}{\frac{\tau}{2} - \epsilon} = O\bigl(\frac{M}{\tau}\bigr)$. Over $L$ levels, the total number of calls to $\textsc{Est}$ with parameters $\epsilon = \frac{\tau}{4}$, and $\gamma = \frac{\delta \tau}{2 M L}$ is $O\bigl(\frac{M}{\tau} \log N\bigr)$, giving the run time.

    \Cref{item:heavy_hitters_1}: suppose $W_{i^\star} \geq \tau$ for some $i^\star \in [N]$. At each level, the half containing $i^\star$ has true weight $\geq W_{i^\star} \geq \tau$, so its estimate is $\geq \tau - \epsilon > \frac{\tau}{2}$. Thus $i^\star$ is tracked to a leaf. At verification, $\hat{W}(\{i^\star\}) \geq W_{i^\star} - \frac{\tau}{4} \geq \frac{3\tau}{4}$, so $i^\star$ (or some other heavy index reached first) is returned.

    \Cref{item:heavy_hitters_2}: if no index has $W_i \geq \frac{\tau}{2}$, then at any leaf $\{i\}$, $\hat{W}(\{i\}) \leq W_i + \frac{\tau}{4} < \frac{3\tau}{4}$, so no leaf passes verification and the algorithm returns $\perp$.

    \emph{Failure probability.} The total number of calls to $\textsc{Est}$ is at most $O\bigl(\frac{M}{\tau} L\bigr)$. By a union bound, $\Pr[\overline{\mcE}] \leq O\bigl(\frac{M}{\tau} L\bigr) \cdot \gamma =\delta$, as required.
\end{proof}

\end{document}